\documentclass[10pt]{article}
\usepackage[preprint]{tmlr} 

\usepackage{amsmath,amssymb,amsthm}

\usepackage{graphicx}
\usepackage{tikz}
\usetikzlibrary{positioning,calc}
\usepackage{booktabs}
\usepackage{multirow}
\usepackage{colortbl}
\usepackage{placeins}

\usepackage{xcolor}
\pagecolor{white}

\usepackage{algorithm}
\usepackage{algpseudocode}

\usepackage{natbib}

\usepackage{hyperref}
\usepackage{url}

\newtheorem{theorem}{Theorem}

\newtheorem{proposition}[theorem]{Proposition}

\theoremstyle{definition}

\theoremstyle{remark}
\newtheorem{remark}[theorem]{Remark}

\title{Conformal Tradeoffs: Operational Profiles Beyond Coverage}

\author{\name Petrus H. Zwart \email PHZwart@lbl.gov \\
      \addr Center for Advanced Mathematics in Energy Research Applications, Berkeley Synchrotron Infrared Structural Biology Program, \& Molecular Biophysics and Integrated Bioimaging Division, 1 Cyclotron Road, Berkeley, CA 94720, USA
      }

\def\month{1}  
\def\year{2026} 
\def\openreview{\url{https://openreview.net/forum?id=XXXX}} 

\begin{document}
\raggedbottom

\maketitle
\begingroup
\renewcommand{\thefootnote}{}
\renewcommand{\theHfootnote}{ai-disclosure}
\footnotetext[1]{\scriptsize Generative AI was used to assist in the preparation and formatting of this document.}
\setcounter{footnote}{0}
\endgroup

\begin{abstract}
Conformal prediction gives exact finite-sample coverage guarantees under
exchangeability, but deployed systems are judged by more than coverage
alone. For a fixed calibrated rule reused over a finite operational
window, stakeholders also care about deployment-facing quantities such as
commitment frequency, deferral, and decisive error exposure. These are
not determined by coverage: calibration choices with similar coverage can
still induce materially different operational profiles.

We study this characterization gap in a scoped setting: binary
split conformal prediction under exchangeability with a fixed deployed
rule. We introduce the Small-Sample Beta Correction (SSBC) which gives finite-sample coverage
semantics for the deployed rule: it inverts the Beta/Beta--Binomial law
governing calibration-conditional coverage to map a user request
$(\alpha^\star,\delta)$ to the least conservative calibration grid point
with calibration-conditional PAC semantics for the realized deployed rule.
Calibrate-and-Audit then fixes the rule by
calibration and uses an independent audit split to estimate the induced
region--class label table, a reusable summary from which deployment-facing
Key Performance Indicators (KPIs) follow by projection. Under this design,
fixed operational rates admit exact finite-sample Binomial inference,
while Beta--Binomial envelopes serve as practical predictive summaries for
future windows. The induced partition also exposes regime boundaries,
Pareto-relevant tradeoffs, and inverse-pricing questions for fixed
downstream conventions.

Simulations validate the SSBC semantics and compare audit-based summaries
with leave-one-out planning proxies; molecular toxicity data provide an
audit-based empirical example, and a solubility case study illustrates scenario
planning once coverage semantics are fixed. 
\end{abstract}


      \section{Introduction: deployment-facing conformal prediction}
\label{sec:introduction}
Conformal prediction provides exact finite-sample coverage guarantees under
exchangeability, but deployment decisions are not made from coverage alone
\citep{vovk2005algorithmic,shafer2008tutorial,AngelopoulosBates2023Gentle}.
Once a classifier is calibrated and reused over a finite operational window,
stakeholders also care about commitment, deferral, and decisive-error
exposure. Those quantities affect throughput and risk, yet they are not
determined by marginal coverage.

We study this gap for binary split conformal prediction with a fixed deployed
rule. Calibration $\tau$ fixes thresholds on score space $\mathcal X$, the thresholds induce a finite region
partition $R_\tau(X)$ and associated fixed labels $Y$, and the deployed system is observed through the joint process
$(R_\tau(X),Y)$. Our central claim is that this region--class structure is the
deployment object of interest: it explains why coverage-matched rules can still
behave differently in practice, and it provides a reusable summary from which
many operational KPIs follow by projection.

Our method has two layers. The Small-Sample Beta Correction (SSBC) maps a user
request $(\alpha^\star,\delta)$ to the least conservative conformal grid point
with the intended finite-sample coverage semantics for the realized deployed
rule. Calibrate--and--Audit then freezes that rule and uses an independent
audit split to estimate deployment-facing rates. Under this design, fixed KPIs
admit exact Binomial inference, while Beta--Binomial envelopes provide
practical planning summaries for future windows.

The exact finite-sample guarantees in the paper are deliberately narrow: they
apply to SSBC's coverage semantics for a fixed deployed rule and to
audit-based inference for pre-declared KPIs at that same fixed rule. By
contrast, calibration sweeps, Pareto views, leave-one-out proxies, and
inverse-pricing analyses are planning tools rather than simultaneous
certification statements.

\begin{figure}[t]
  \centering
  \includegraphics[width=\textwidth]{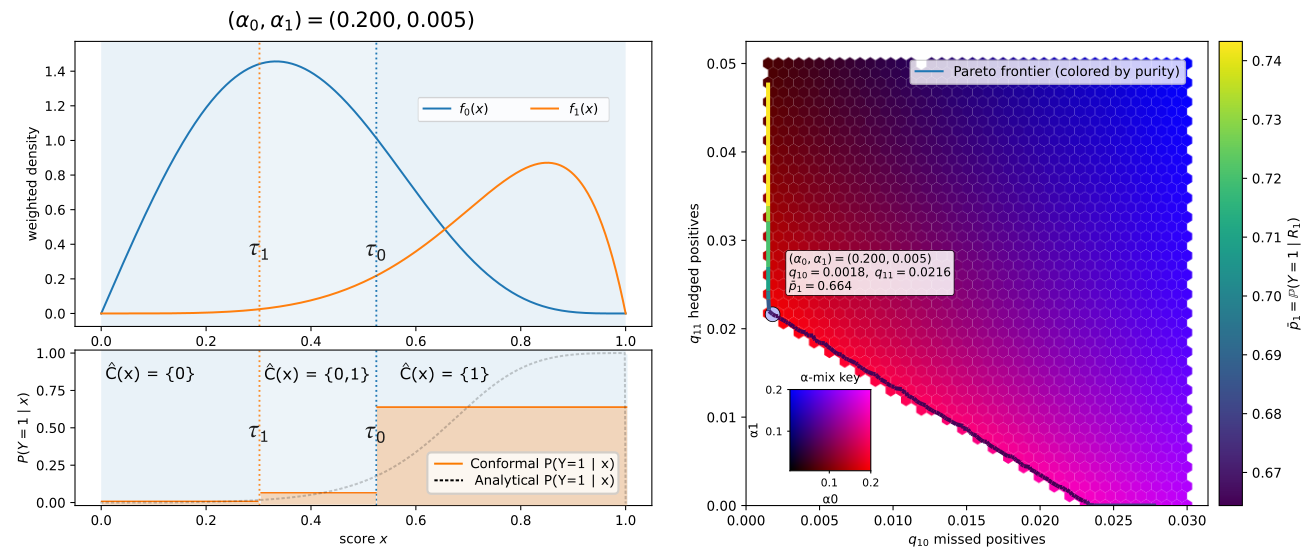}
  \caption{\textbf{Operational view of a calibrated conformal rule.} Calibration
  fixes a region partition; auditing the resulting region--class table yields a
  reusable operational summary; sweeping calibration settings traces feasible
  trade-offs.}
  \label{fig:hook}
\end{figure}

\subsection{Contributions}
\label{subsec:contributions}
In the scoped setting of binary split conformal prediction under
exchangeability, the paper makes three contributions.

\paragraph{(1) Operational structure beyond coverage.}
The calibration-induced partition constrains deployment behavior: changing
thresholds reallocates mass across a small set of region types, so operational
KPIs are coupled rather than independently tunable. This yields attainable
sets, Pareto-relevant regimes, and a natural interface for asking which
downstream action conventions are justified by the audited region-wise label
composition.
\paragraph{(2) SSBC as a coverage-semantics anchor.}
SSBC inverts the exact rank/Beta law to map $(\alpha^\star,\delta)$ to the
least conservative deployed grid point whose realized coverage satisfies a
calibration-conditional PAC-style statement. It turns a nominal request into an
explicit finite-sample semantic claim about the realized conformal rule.
\paragraph{(3) Calibrate--and--Audit for deployment KPIs.}
The region--class label table is the reusable audit object for a fixed rule.
From it one can derive commitment, deferral, singleton-error, and related KPI
rates by projection. Exact Binomial inference applies pointwise to fixed
audit-based KPIs, while leave-one-out proxies are used only for exploratory
planning when an audit split is unavailable.

\subsection{Relation to prior work}
\label{subsec:positioning}
Conformal prediction constructs set-valued predictors with finite-sample
coverage guarantees under exchangeability
\citep{vovk2005algorithmic,shafer2008tutorial,AngelopoulosBates2023Gentle}.
Conformal risk control (CRC) extends this perspective by selecting thresholds
to satisfy user-specified \emph{scalar} risk constraints under related
assumptions \citep{Angelopoulos2024CRC,bates2021distribution}, and other work
studies stronger conditional notions of validity or connects conformal methods
to downstream decision-making
\citep{tibshirani2019conformal,Fannjiang2022FeedbackCovariateShift,
Bian2023TrainingConditional,Gibbs2025Conditional,Lekeufack2023CDT,Kiyani2025DTFCP}.
These studies typically ask how to guarantee coverage, a chosen scalar risk
functional, or decision quality under a specified objective. Our emphasis is
different: once a conformal rule is fixed, what operational behavior should one
expect from it, and how much of that behavior is already determined by
calibration geometry rather than by coverage alone?

Recent Human-Computer Interaction (HCI)-oriented work raises a related
deployment concern: conformal sets can be a ``murky'' interface when validity
summaries do not translate cleanly into action-relevant
consequences~\citep{hullman2025conformalpredictionhumandecision}. Our
region--class label table view puts that issue on a concrete footing: it
audits how mass falls across region types and labels, exposing how often a rule
commits, hedges, or abstains, and where decisive errors arise. This makes clear
how coverage-matched rules can still behave differently in deployment. We do
not replace validity guarantees; we clarify what they do and do not imply for
deployment behavior. Once constructed, the table supports many KPIs and policy
projections without re-running calibration.

Closest in spirit to our viewpoint is the inverse CRC formulation of
\citet{zhouzhu2026inversecrc}, who trace certified miscoverage--regret
trade-offs for predict-then-optimize pipelines, and related work on conformal
efficiency optimizes set-size functionals subject to validity constraints
\citep{yangkuchibhotla2021}. Our objective is different: rather than introducing
another scalar criterion, we make the \emph{deployment interface} itself the
primitive. A calibration choice $\theta$ fixes a finite conformal partition,
and the induced region--class label table is the minimal reusable summary that
determines \emph{many} operational KPIs by linear projection under any fixed
region-based policy; this yields an \emph{operational profile} (rate-vector)
map $\theta \mapsto \mathbf{r}(\theta)$ and an attainable set of behaviors that is
geometrically constrained, so coverage-matched rules can still induce
qualitatively different deployment regimes. Methodologically, we enforce a
disciplined two-stage workflow: calibration sweeps and KPI estimates on the audit 
split, or LOO surrogates when an audit split is unavailable, are used only to 
expose feasible regimes and trade-offs, while exact finite-sample guarantees 
are asserted \emph{pointwise} for pre-declared KPIs at a fixed deployed rule 
using an independent audit split via Binomial inference (with Beta--Binomial 
envelopes as planning summaries for finite windows). SSBC plays a complementary 
role as a \emph{semantics anchor}, mapping a user request $(\alpha^\star,\delta)$ 
to a concrete calibration grid point so the operational map is navigated relative 
to an explicit finite-sample coverage semantics for the realized rule.

\subsection{Roadmap}
\label{subsec:roadmap}
The remainder of the paper is organized as follows: Section~\ref{sec:setting} defines the fixed deployed object and the
region--class label table, Section~\ref{sec:operational-quantities} presents
SSBC and Calibrate--and--Audit, Section~\ref{sec:geometric_manifolds} gives the
binary geometric interpretation, Section~\ref{sec:experiments:setup} reports
simulation, molecular toxicology, and solubility scenario planning, and 
Section~\ref{sec:discussion} concludes. Technical derivations and extended
examples are deferred to the appendices.
      
      \section{Setting and notation: calibration-conditional viewpoint}
\label{sec:setting}

We study a \emph{single deployed classifier}: a scoring model is trained once,
treated as fixed, and then calibrated once by split conformal prediction.
Randomness therefore comes from the calibration draw and future deployment
examples, not from retraining or repeated recalibration. To evaluate
deployment-facing quantities beyond coverage, we reserve an exchangeable audit
split that is never used to choose thresholds.

\subsection{Data splits and exchangeability assumptions}

Let $\mathcal{D}_{\mathrm{train}}$ denote the training data used to fit a base
score model. After training, the score function is treated as fixed. Let
\[
\mathcal{D}_{\mathrm{cal}}=\{(X_i^{c},Y_i^{c})\}_{i=1}^{n_{\mathrm{cal}}},
\qquad
\mathcal{D}_{\mathrm{audit}}=\{(X_i^{a},Y_i^{a})\}_{i=1}^{n_{\mathrm{audit}}}.
\]
The calibration split $\mathcal{D}_{\mathrm{cal}}$ is used once to set the
deployed thresholds, while $\mathcal{D}_{\mathrm{audit}}$ is reserved for
post-calibration evaluation of fixed-rule KPIs.

We assume calibration, audit, and future deployment points are jointly
exchangeable conditional on the fixed trained score model. We focus on binary
classification because it exposes the geometry and deployment trade-offs most
cleanly.

\subsection{Scores and calibration thresholds}
\label{subsec:scores-thresholds}

Let $\mathcal{Y}=\{1,\ldots,K\}$, and let
$s:\mathcal{X}\times\mathcal{Y}\to\mathbb{R}$ be a nonconformity score
\citep{shafer2008tutorial,lei2018distribution}. In the experiments we use
$s(x,y)=1-P(y\mid x)$, but the setup below does not depend on that choice. 

Given $\mathcal{D}_{\mathrm{cal}}$, compute calibration scores
\[
S_i := s(X_i^c,Y_i^c), \qquad i=1,\ldots,n_{\mathrm{cal}},
\]
and let $S_{(1)}\le\cdots\le S_{(n_{\mathrm{cal}})}$ denote their order
statistics. Split conformal calibration selects
\[
\tau := S_{(k)},
\]
where $k\in\{1,\ldots,n_{\mathrm{cal}}\}$ is determined by the requested
miscoverage level. Equivalently, one may index the same grid by
$u:=n_{\mathrm{cal}}+1-k$, so the deployed grid level is
$\alpha_{\mathrm{grid}}=u/(n_{\mathrm{cal}}+1)$.

We use class-conditional split conformal
\citep{Vovk2012PAC,vovk2012}, so thresholds $\tau_y$ are computed separately
within each class. For a fixed calibration draw, this yields a threshold vector
$\tau=(\tau_1,\ldots,\tau_K)$ that remains fixed during deployment.

\subsection{Regions, observability, and policies}
\label{subsec:regions}

Once thresholds are fixed, they induce a finite region label. Writing
\[
s(x) := (s(x,1),\ldots,s(x,K)),
\qquad
\tau := (\tau_1,\ldots,\tau_K),
\]
the region map is
\[
R_\tau(x)
:=
\big(
\mathbf{1}\{s(x,1)\le\tau_1\},\ldots,
\mathbf{1}\{s(x,K)\le\tau_K\}
\big)
\in\{0,1\}^K.
\]
Thus calibration fixes a finite partition of score space. The partition itself
depends only on $\tau$, while the amount of mass carried by each region depends
on the deployment distribution. This partition is the paper's operational
interface (Figure~\ref{fig:regions}).

\begin{figure}[t]
  \centering
  \includegraphics[width=\linewidth]{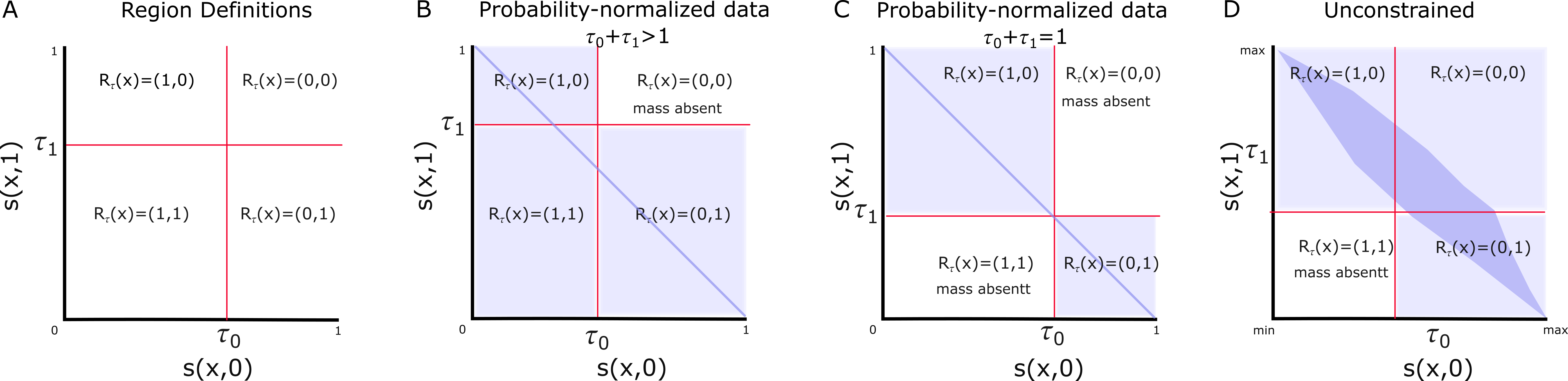}
  \caption{\textbf{Thresholds induce a finite region partition.} In the binary
  probability-normalized setting, the score support determines which regions can
  carry mass under a fixed pair of thresholds.}
  \label{fig:regions}
\end{figure}

Deployment is observed through the joint process $(R_\tau(X),Y)$. To make
downstream actions explicit, we also allow a fixed deployment policy
$\pi:\mathcal R\to 2^{\mathcal Y}$ and write
\[
\hat C_\pi(x) := \pi(R_\tau(x)).
\]
The standard conformal predictor is the set-inclusion policy
\[
\pi_{\mathrm{SI}}(r) := \{y\in\mathcal{Y} : r_y=1\},
\qquad
\hat C_{\pi_{\mathrm{SI}}}(x)=\{y\in\mathcal{Y} : s(x,y)\le\tau_y\}.
\]
As an example beyond set inclusion, consider a region-triggered commit policy:
in the binary case, choose a trigger set $A\subseteq\{0,1\}^2$ and an action
$a\in\{0,1\}$, then commit to $\{a\}$ when $R_\tau(x)\in A$ and abstain
otherwise. This simple template makes the calibration/action separation
concrete: calibration fixes the region partition, while policy chooses how to
act on the realized region. Additional binary projection masks can be found 
in Appendix~\ref{app:explicit_region_indicators}.

\begin{figure}[t]
  \centering
  \includegraphics[width=\linewidth]{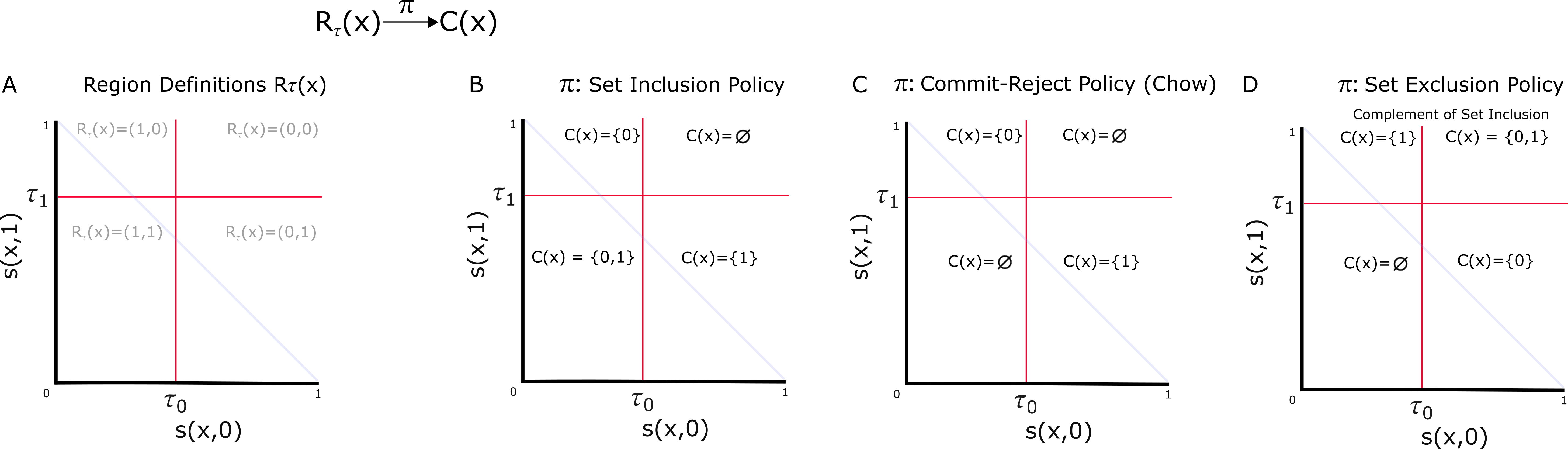}
  \caption{\textbf{Deployment policies as projections on a fixed region
  structure.} Different region-based policies act on the same calibrated
  partition but induce different reported outputs.}
  \label{fig:policy-projection}
\end{figure}

\subsection{Auditing primitives: region--class label tables and policy projections}
\label{sec:setting:targets}

Conditional on $\mathcal D_{\mathrm{cal}}$, the deployed rule is fixed, and the
primitive observable for auditing is the pair $(R_{\tau(\theta)}(X),Y)$ on an
exchangeable labeled split. Here $\theta\in\Theta$ indexes thresholds
$\tau(\theta)$ and hence the region map $R_{\tau(\theta)}$. Two linked objects
will be used throughout the paper.

\paragraph{(A) Region--class label joint table.}
Define the calibration-conditional joint probabilities
\[
p_{r,y}(\theta)
:=
\mathbb P\!\big(R_{\tau(\theta)}(X)=r,\ Y=y \mid \mathcal D_{\mathrm{cal}}\big),
\qquad (r,y)\in\mathcal R\times\mathcal Y.
\]
These are fixed but unknown constants for the deployed rule. The audit split
estimates them through
\[
K^{\mathrm{audit}}_{r,y}(\theta)
:=
\sum_{i=1}^{n_{\mathrm{audit}}}
\mathbf{1}\{R_{\tau(\theta)}(X_i^{a})=r,\ Y_i^{a}=y\},
\qquad
\widehat p^{\,\mathrm{audit}}_{r,y}(\theta)
=\frac{K^{\mathrm{audit}}_{r,y}(\theta)}{n_{\mathrm{audit}}}.
\]
Because the cells $(r,y)$ form a finite partition, the table sums to one and
changing $\theta$ reallocates mass across regions. Those conservation
constraints drive the geometric coupling studied later.

\paragraph{(B) Policy-specific indicators.}
Fix a deployment policy $\pi$. Many KPIs of interest can be written as a Bernoulli
indicator
\[
I_\ell := g_\ell(R_{\tau(\theta)}(X),Y;\pi)\in\{0,1\},
\qquad
p_\ell(\theta)
:=
\mathbb P\!\big(I_\ell=1\mid \mathcal D_{\mathrm{cal}}\big).
\]
For example,
$g_{\mathrm{abs}}(r,y;\pi)=\mathbf 1\{\pi(r)=\varnothing\}$ and
$g_{\mathrm{err}}(r,y;\pi)=\mathbf 1\{|\pi(r)|=1,\ y\notin\pi(r)\}$.
The corresponding audit count is
\[
K^{\mathrm{audit}}_{\ell}(\theta)
:=
\sum_{i=1}^{n_{\mathrm{audit}}}
g_\ell(R_{\tau(\theta)}(X_i^{a}),Y_i^{a};\pi),
\qquad
\widehat p^{\,\mathrm{audit}}_{\ell}(\theta)
=\frac{K^{\mathrm{audit}}_{\ell}(\theta)}{n_{\mathrm{audit}}}.
\]
Thus the region--class label table is the reusable audit object, and KPI rates
are its projections.

\paragraph{Projection identity.}
\label{sec:setting:linear}
The table and KPI views are linked by linearity:
\begin{equation}
\label{eq:projection_identity}
K^{\mathrm{audit}}_{\ell}(\theta)
=
\sum_{r\in\mathcal R}\sum_{y\in\mathcal Y}
g_\ell(r,y;\pi)\,K^{\mathrm{audit}}_{r,y}(\theta),
\qquad
\widehat p^{\,\mathrm{audit}}_{\ell}(\theta)
=\frac{K^{\mathrm{audit}}_{\ell}(\theta)}{n_{\mathrm{audit}}}.
\end{equation}
Taking expectation gives
\[
p_\ell(\theta)
=
\sum_{r\in\mathcal R}\sum_{y\in\mathcal Y}
g_\ell(r,y;\pi)\,p_{r,y}(\theta).
\]
In words: estimate the table once, then obtain policy-level KPIs by projection.

\paragraph{Worked example: coverage as a projection of the region--class label table.}
Under set inclusion, coverage is the sum of the region--class label cells where
the true label lies in the predicted set. Consider $\mathcal Y=\{0,1\}$ and the
set-inclusion policy $\pi_{\mathrm{SI}}$. For fixed thresholds, each input falls
into one of four regions $\mathcal R=\{r_{10},r_{11},r_{01},r_{00}\}$, and the
primitive auditing object is $(R_\tau(X),Y)\in\mathcal R\times\mathcal Y$. The
corresponding region--class label table is
\[
P(\theta):
\begin{array}{c|ccc}
 & y=0 & y=1 & \hat C_{\pi_{\mathrm{SI}}}(X) \\ \hline
r_{10}=(1,0) & \mathbf{p_{10,0}(\theta)} & p_{10,1}(\theta) & \{0\}\\
r_{11}=(1,1) & \mathbf{p_{11,0}(\theta)} & \mathbf{p_{11,1}(\theta)} & \{0,1\}\\
r_{01}=(0,1) & p_{01,0}(\theta) & \mathbf{p_{01,1}(\theta)} & \{1\}\\
r_{00}=(0,0) & p_{00,0}(\theta) & p_{00,1}(\theta) & \varnothing
\end{array}
\qquad\text{with}\qquad
\sum_{r\in\mathcal R}\sum_{y\in\mathcal Y}p_{r,y}(\theta)=1.
\]
Under $\pi_{\mathrm{SI}}$, the event $\{Y\in \hat C_{\pi_{\mathrm{SI}}}(X)\}$
corresponds to the four bold cells $(r_{10},0)$, $(r_{11},0)$, $(r_{11},1)$,
and $(r_{01},1)$, hence
\[
p_{\mathrm{cov}}(\theta)
=
p_{10,0}(\theta)+p_{11,0}(\theta)+p_{11,1}(\theta)+p_{01,1}(\theta)
=
1-\big(p_{10,1}(\theta)+p_{01,0}(\theta)+p_{00,0}(\theta)+p_{00,1}(\theta)\big).
\]
The calibration choices $(\alpha_0^\star,\alpha_1^\star)$ determine the
class-conditional miscoverage primitives
$p_{10,1}(\theta)+p_{00,1}(\theta)$ for $Y=1$ and
$p_{01,0}(\theta)+p_{00,0}(\theta)$ for $Y=0$, while the decomposition of these
totals into wrong-singleton errors versus abstentions is dictated by the
deployment distribution. This makes explicit that (i) coverage is a projection,
(ii) once thresholds are fixed it is determined by how mass is redistributed
across regions, and (iii) deployment-facing consequences are not determined by
the calibration targets alone.

The same ``sum selected cells'' structure applies to abstention/deferral,
decisiveness, and decisive error exposure under any fixed policy. Appendix~
\ref{app:explicit_region_indicators} records additional binary projection masks
and related bookkeeping.

      \section{Operational quantities as first-class objects}
\label{sec:operational-quantities}

Deployment behavior is summarized by operational rates induced by a calibrated
partition together with a fixed deployment policy~$\pi$. The purpose of this
section is to separate what is certifiable for a fixed deployed rule from what
is useful only for planning across candidate rules. We first fix the coverage
semantics of the deployed rule through SSBC, then turn to audit-based inference
and comparative planning at that fixed rule. We keep three layers distinct:
geometry (the induced regions), policy (the projection from regions to
outputs), and rates (the auditable deployment frequencies).

\subsection{Calibrate--and--Audit}
\label{sec:operational-quantities:calibrate_audit}

Beyond marginal coverage, finite-sample evaluation of operational rates
requires an independent audit split. If $\mathcal D_{\mathrm{cal}}$ is reused
both to choose thresholds and to evaluate downstream indicators, the resulting
event counts are no longer Bernoulli draws conditional on a fixed rule.
Calibrate--and--Audit avoids this coupling by freezing thresholds on
$\mathcal D_{\mathrm{cal}}$ and evaluating operational quantities on
$\mathcal D_{\mathrm{audit}}$ only.

When an audit split is unavailable, we use leave-one-out (LOO) recalibration as
a practical proxy. Throughout the paper, that proxy is used only for
exploration and scenario planning, not for the exact fixed-rule guarantees.
The structural issue is that calibration reuse makes the threshold random with
respect to the same observations used for evaluation, inducing dependence
between selection and KPI counts rather than conditionally i.i.d.\ Bernoulli
trials. Appendix~\ref{app:single_sample_structural_coupling} gives the
covariance argument and the details of the LOO construction.

\subsection{SSBC as a calibration navigation coordinate}
\label{sec:operational-quantities:ssbc_navigation}

We use coverage as the semantic anchor for navigating the calibration grid. Split conformal
calibration lives on a finite grid indexed by
$u\in\{1,\ldots,n_{\mathrm{cal}}\}$ with
$\alpha_{\mathrm{grid}}=u/(n_{\mathrm{cal}}+1)$; see
Section~\ref{subsec:scores-thresholds}. SSBC maps $(\alpha^\star,\delta)$ to the
least conservative admissible index $u^\star(\alpha^\star,\delta)$ satisfying
\[
\mathbb{P}_{\mathcal{D}_{\mathrm{cal}}}\!\left(
\mathbb{P}\!\big(Y \in \hat C(X)\mid \mathcal{D}_{\mathrm{cal}}\big)\ge 1-\alpha^\star
\right)\ge 1-\delta.
\]
Thus SSBC turns a semantic request into a concrete deployed threshold choice. In the binary class-conditional setting this yields
\[
(u_0^\star,u_1^\star)
=
\big(u^\star(\alpha_0^\star,\delta_0),\,u^\star(\alpha_1^\star,\delta_1)\big),
\]
which determine $(\tau_0,\tau_1)$ and therefore the calibration setting
$\theta$. Although the user-facing request is four-dimensional,
$(\alpha_0^\star,\delta_0,\alpha_1^\star,\delta_1)$, calibration collapses it
to a low-dimensional navigation coordinate on the deployed grid. The later
planning objects are therefore indexed by a coverage-semantically meaningful
choice of deployed rule rather than by an abstract threshold sweep.

\subsection{Audit-based predictive envelopes for future windows}
\label{sec:operational-quantities:certified_envelopes}

Under Calibrate--and--Audit, conditional on $\mathcal D_{\mathrm{cal}}$ the
deployed rule is fixed and audit set is exchangeable with future deployment
points. Therefore, for any fixed KPI indicator $I_\ell$,
\[
K^{\mathrm{audit}}_{\ell}(\theta)\mid \mathcal D_{\mathrm{cal}}
\sim
\mathrm{Binomial}\!\big(n_{\mathrm{audit}}, p_\ell(\theta)\big),
\]
and for a future window of size $m$,
\[
K^{m}_{\ell}(\theta)\mid \mathcal D_{\mathrm{cal}}
\sim
\mathrm{Binomial}\!\big(m, p_\ell(\theta)\big).
\]
These Binomial laws support exact fixed-rule inference for the latent rate
$p_\ell(\theta)$; for example, Clopper--Pearson intervals \citep{clopper1934use} are valid from the
audit count $K^{\mathrm{audit}}_{\ell}(\theta)$. For finite future windows we
use Beta--Binomial envelopes as practical planning summaries
\citep{Skellam1948BinomialMixture,JohnsonKotzBalakrishnan1997,Gelman2013BDA}. 

\subsection{Rate vectors, attainable operational sets, and Pareto filtering}
\label{sec:operational-quantities:rates_interface}

The guarantees above are \emph{pointwise}: they certify a pre-declared KPI at a
fixed deployed rule, not an adaptively selected member of a sweep. Planning is
comparative. For a calibration family $\Theta$, a fixed policy $\pi$, and a KPI
list, each setting $\theta\in\Theta$ determines an operational profile
\[
\mathbf p(\theta)=\big(p_1(\theta),\ldots,p_L(\theta)\big),
\]
whose coordinates are obtained by projection from the same region--class label
table via \eqref{eq:projection_identity}. Sweeping $\theta$ traces the
attainable set
\[
\mathcal V(\Theta;\pi)
:=
\{\mathbf p(\theta):\theta\in\Theta\}
\subset \mathbb R^L.
\]
Because the underlying table obeys conservation constraints, changing $\theta$
reallocates mass across regions rather than tuning rates independently.

To summarize planning-relevant regimes without committing to a scalar objective,
we use an oriented Pareto filter. After declaring which KPIs are preferable to
increase or decrease, a setting is \emph{nondominated} if no other setting is
at least as good in every oriented coordinate and strictly better in one. These
fronts are planning objects; exact finite-sample guarantees remain pointwise
for fixed settings and fixed KPIs. Section~\ref{sec:geometric_manifolds} now
makes precise why these attainable sets are geometrically constrained and why
sweeping SSBC-indexed settings produces coupled, regime-dependent trade-offs.
Any later scalar objective that is monotone in the chosen orientation must
attain its optimum on this front \citep{Miettinen1999NonlinearMultiobjective}. The follow-on question is inverse pricing: which
downstream cost ratios justify a fixed action convention at a Pareto-relevant
regime? Further derivations are deferred to Appendices~\ref{app:ssbc}
and~\ref{sec:alpha_vs_delta}.

      \section{Consequences of a fixed conformal partition in the binary case}
\label{sec:geometric_manifolds}
\noindent
This section isolates the main binary geometric fact behind the operational
trade-offs. Once thresholds are fixed, deployment behavior is governed by the
joint region--class table $(R_\tau(X),Y)$. Under probability-normalized scores,
the threshold pair $(\tau_0,\tau_1)$ cannot move operational KPIs
independently; it reallocates mass across a small number of region types.
Equivalently, this section makes precise the conservation-constraint intuition
behind the attainable sets and Pareto filtering introduced in
Section~\ref{sec:operational-quantities:rates_interface}. In the binary case
this yields three practical implications: coupled feasibility, region-wise
policy dependence, and decision optimality determined by the audited table
rather than by coverage alone.
Appendix~\ref{app:geometry_binary_partition} contains the full construction and
proofs, while Appendix~\ref{app:cost_pricing_envelopes} develops the
decision-theoretic consequences.

\subsection{Binary conformal geometry and regime boundaries}
\label{sec:geometric_manifolds:binary_conformal_partitions}

Let $\mathcal Y=\{0,1\}$ and fix class-conditional thresholds
$\tau=(\tau_0,\tau_1)$. Once calibrated, the object that governs operational
characteristics in deployment is the region label
\[
R_\tau(x)
:=
\big(\mathbf{1}\{s(x,0)\le \tau_0\},\ \mathbf{1}\{s(x,1)\le \tau_1\}\big)
\in\{0,1\}^2,
\]
with outcomes $10,11,01,00$, which under the set inclusion policy $\pi_{\mathrm{SI}}$ correspond to singleton-$0$, hedge, 
singleton-$1$, and abstention, respectively.

\paragraph{Probability-normalized scores induce a regime boundary.}
For probability-normalized scores $s(x,y)=1-P(y\mid x)$ we have
$s(x,0)+s(x,1)=1$. Hence a point cannot satisfy both class thresholds unless
$\tau_0+\tau_1\ge 1$, and it cannot violate both unless $\tau_0+\tau_1\le 1$.
Consequently:
\begin{itemize}
  \item \textbf{Hedging regime ($\tau_0+\tau_1>1$):} $11$ may occur and $00$ cannot
  (singletons + hedges; no abstention).
  \item \textbf{Rejection regime ($\tau_0+\tau_1<1$):} $00$ may occur and $11$ cannot
  (singletons + abstention; no hedges).
  \item \textbf{Boundary ($\tau_0+\tau_1=1$):} only $10$ and $01$ occur; under
  $\pi_{\mathrm{SI}}$ outputs are always singletons.
\end{itemize}
Crossing $\tau_0+\tau_1=1$ therefore changes which outcome types are even
feasible. This deterministic boundary is the simplest geometric reason that
changing thresholds can produce qualitatively different operational regimes.

\paragraph{Cross-threshold coupling within regimes.}
Within the hedging regime ($\tau_0+\tau_1>1$), the diagonal support is partitioned
into three contiguous intervals. Parameterize by $u=s(x,0)$ (so $s(x,1)=1-u$):
\[
u\in[0,1-\tau_1)\ \Rightarrow\ R_\tau(x)=10,\qquad
u\in[1-\tau_1,\tau_0]\ \Rightarrow\ R_\tau(x)=11,\qquad
u\in(\tau_0,1]\ \Rightarrow\ R_\tau(x)=01.
\]
Thus region boundaries are governed by \emph{opposing} thresholds: changing
$\tau_1$ moves the boundary controlling singleton-$0$ mass, while changing
$\tau_0$ moves the boundary controlling singleton-$1$ mass. Thresholds are
therefore mass-reallocation boundaries, not independent class-wise knobs.

For planning, the key consequence is that sweeping calibration settings can
cross regime boundaries and induce discontinuous changes in which outputs are
even feasible. This is the geometric source of the coupled attainable sets and
Pareto fronts seen later in the experiments.

\subsection{Region observability and interface-relative decision optimality}
\label{sec:geometric_manifolds:posterior_view}

Fix $\tau$ and treat the region label $R_\tau(X)\in\{0,1\}^2$ as the deployed
observable. The joint region--class label probabilities
\[
p_{r,y}
:=
\mathbb P\!\left(R_\tau(X)=r,\ Y=y\mid\mathcal D_{\mathrm{cal}}\right),
\qquad r\in\{0,1\}^2,\ y\in\{0,1\},
\]
fully characterize the information available to any downstream rule that uses
only this interface. Write $p_r:=p_{r,0}+p_{r,1}$ for region mass and define the
within-region label frequency
\[
\eta_r:=\Pr(Y=1\mid R_\tau(X)=r,\ \mathcal D_{\mathrm{cal}})
=\frac{p_{r,1}}{p_r}\quad (p_r>0).
\]
Because the interface takes finitely many values, both evaluation and
decision-making reduce to functions of the same table $\{p_{r,y}\}$ (equivalently
$\{p_r,\eta_r\}$). In particular, a region-based action convention is justified
by the \emph{within-region} label composition, not by marginal coverage alone
and not by the set-valued output alone.

Formally, let $\mathcal A$ be an action set and let $L(a,y)$ denote the loss of
taking action $a\in\mathcal A$ when $Y=y$. A region-based policy
$\tilde\pi:\{0,1\}^2\to\mathcal A$ is \emph{decision-optimal relative to the
conformal interface} if, for every populated region $r$,
\[
\tilde\pi(r)\in\arg\min_{a\in\mathcal A}
\mathbb E\!\left[L(a,Y)\mid R_\tau(X)=r,\ \mathcal D_{\mathrm{cal}}\right].
\]
In words, once the deployed interface reveals only the region label $r$, the
optimal action is the one with smallest expected loss under the label mix within
that region. This optimality is determined region-wise by the within-region label frequency $\eta_r$.

For the main text, the key implication is simply that the same audited
region--class table used for operational planning also determines whether a
fixed region-based convention is rational under a stated cost model. The full
Chow-style inverse-pricing analysis is deferred to
Appendix~\ref{app:cost_pricing_envelopes}, but two qualitative consequences are
worth flagging here: the resulting inverse-pricing envelope is polyhedral in
cost-ratio coordinates, and coverage-matched settings can induce different, even
disjoint, pricing envelopes because the region composition changes.

      \section{Results: evidence for coverage semantics, audit behavior, and planning}
\label{sec:applications}
\label{sec:experiments:setup}

This section keeps the empirical story at the same scope as the theory. We
validate SSBC's coverage semantics in simulation, use a toxicology dataset as an audit-based
example of fixed-rule operational evaluation, and close with a solubility
scenario-planning illustration in which the model is fixed and only the
calibration layer is varied.

\subsection{Numerical simulations: coverage semantics first, operational summaries second}

\subsubsection{Coverage: numerical realization of SSBC guarantees}
\label{sec:applications:simulations}
\label{sec:applications:ssbc}
\label{sec:coverage_control}

We first isolate the finite-sample law underlying SSBC. Calibration
nonconformity scores are drawn i.i.d.\ from a continuous heavy-tailed reference
distribution, and we compare nominal split conformal, a one-sided DKWM
correction \citep{Massart1990DKWM,dvoretzky1956asymptotic,Vovk2012PAC}, and SSBC
at $(\alpha^\star,\delta)=(0.10,0.10)$ over a finite deployment window of size
$m_{\mathrm{infer}}=100$.

Nominal split conformal under-controls calibration-conditional risk in this
finite-window view, whereas DKWM enforces the target through strong
conservatism. Across representative calibration sizes, nominal split conformal
produces violation probabilities far above the requested level, DKWM drives the
same probabilities well below target by substantial conservatism, and SSBC
tracks the intended finite-window semantics much more closely up to unavoidable
grid effects. For example, at $n_{\mathrm{cal}}=100$ the observed coverage violation
rate is $0.4075$ for nominal split conformal, $0.0004$ for DKWM, and $0.0960$
for SSBC, close to the target $\delta=0.10$. The full calibration-size grid
with theoretical and observed violation rates is reported in Appendix~\ref{app:ssbc},
Table~\ref{tab:calibration_conditional_violations}.

\subsubsection{Operational rate summaries: LOO versus two-sample audit reference}
\label{sec:applications:simulations_envelopes}
\label{sec:applications:two_sample_reference}

We next study operational quantities beyond coverage. The goal is to compare
the two-sample Calibrate--and--Audit constructed Beta--Binomial predictive
summaries to a single-sample LOO surrogate intended for feasibility
exploration.

\paragraph{Synthetic probability model.}
We use a controlled binary probability model. Each draw produces
$Y\in\{0,1\}$ with $\mathbb P(Y=1)=p_{\mathrm{class}}$ for
$p_{\mathrm{class}}\in\{0.10,0.50\}$, and
\[
P_1 \mid Y=1 \sim \mathrm{Beta}(a,b),
\qquad
P_1 \mid Y=0 \sim \mathrm{Beta}(2,7),
\]
with $(a,b)\in\{(4,3),(9,3)\}$. We then set $(P_0,P_1)=(1-P_1,P_1)$ and use
class-conditional scores $S_y:=1-P_y$.

\paragraph{Two-sample reference and LOO surrogate.}
For each configuration, we draw independent datasets $\mathcal D_1$ and
$\mathcal D_2$ of size $N=500$, calibrate on $\mathcal D_1$ with
SSBC-adjusted thresholds at $(\alpha,\delta)=(0.10,0.10)$, freeze the rule,
and evaluate operational indicators on $\mathcal D_2$. This yields the
two-sample audit reference. As a single-sample surrogate, we recompute
thresholds leave-one-out, pool the resulting indicators, and map them to
planning envelopes, optionally widened by controlled pessimization
(Appendix~\ref{app:single_sample_structural_coupling}).

Figure~\ref{fig:operational_interval_alignment} shows that, in these simulated
geometries, LOO-based envelopes align closely with the two-sample reference,
while inflation widens intervals without shifting centers. This supports LOO as
a practical planning proxy when an explicit audit split is unavailable, while
keeping the main operational evidence tied to the two-sample design.

\begin{figure}[H]
    \centering
    \includegraphics[width=0.85\textwidth]{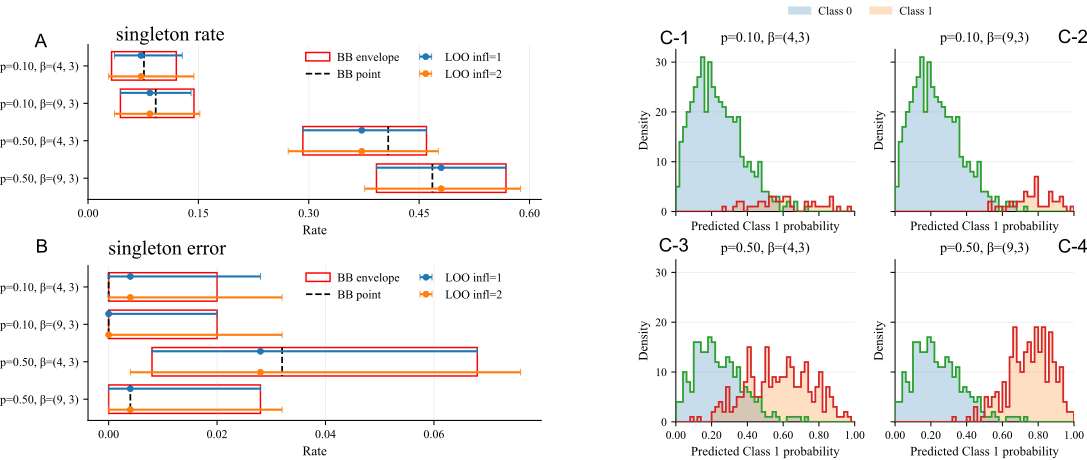}
    \caption{Operational rate envelopes and score geometry.
    \textbf{(A)} Singleton rate and \textbf{(B)} singleton error for two class
    prevalences and two class~1 generating distributions, with class~0 drawn
    from $\mathrm{Beta}(2,7)$. Red rectangles denote the two-sample
    Beta--Binomial future-window summary and the dashed vertical line its
    center; blue and orange intervals show leave-one-out envelopes under two
    inflation levels. \textbf{(C1--C4)} Histograms of predicted class~1
    probability by true class make explicit how score geometry shapes singleton
    mass, singleton error, and the asymmetry of the envelopes.}
    \label{fig:operational_interval_alignment}
\end{figure}

\subsection{Tox21: an empirical example of the audit-based operational view}
\label{sec:applications:tox21}
\label{sec:applications:envelopes}

Tox21 \citep{Mayr2016DeepTox,Huang2016Tox21} stress-tests the framework under
severe class imbalance, where class-conditional calibration counts can be small.
Across twelve endpoints and 100 random train/calibration/audit splits, we
compare nominal split conformal, DKWM, and SSBC at $(\alpha,\delta)=(0.10,0.10)$.
Dataset composition and representative endpoint-level operational tables are
deferred to Appendix~\ref{app:tox21}.

The aggregate calibration-conditional picture follows the same trend as the numerical simulation:
nominal split conformal again shows elevated coverage violation in this
conditional small-$n$ regime, DKWM suppresses violations through strong
conservatism, and SSBC lands much closer to the intended finite-sample
semantics while avoiding DKWM's excess inflation. The set-size trade-off is
similarly clear: mean set size is $1.41$ for nominal split conformal, $1.78$
for DKWM, and $1.54$ for SSBC, with singleton frequency ordered in the opposite
direction. The aggregated coverage and set-size summary is reported in
Appendix~\ref{app:tox21}, Table~\ref{tab:tox21-summary}.

\subsubsection{Operational summaries on an independent evaluation split}

With coverage semantics fixed by SSBC at $(\alpha,\delta)=(0.10,0.10)$, we examine the induced region--class
summaries on an independent audit split. The point of the endpoint-level table
is different from the aggregate coverage table above: it shows what the fixed
rule actually does in deployment-facing KPI terms and how closely the LOO proxy
tracks that held-out operational picture.

\begin{table}[ht]
    \centering
    \caption{\textbf{Representative Tox21 endpoint: SR-MMP.} Joint rates are
    normalized by the endpoint test-set size. The $\mathcal{D}_{1}$ LOO columns provide planning
    summaries (Point Estimate PE and 95\% Prediction Interval PI) from calibration data, while the $\mathcal{D}_2$ columns report
    the independent audit reference and its predictive interval.}
    \label{tab:sr-mmp-operational}
    \footnotesize
    \begin{tabular}{llcccc}
    \hline
    Operational quantity & Class & LOO PE $\mathcal{D}_{1}$ & LOO 95\% PI $\mathcal{D}_{1}$ & PE $\mathcal{D}_{2}$ & BB 95\% PI $\mathcal{D}_{2}$\\
    \hline
    Singleton rate & Class 0 & $0.662$ & $[0.604,0.718]$ & $0.651$ & $[0.615,0.686]$ \\
                    & Class 1 & $0.130$ & $[0.092,0.173]$ & $0.117$ & $[0.094,0.142]$ \\
    Doublet rate & Class 0 & $0.163$ & $[0.121,0.211]$ & $0.201$ & $[0.172,0.231]$ \\
                 & Class 1 & $0.045$ & $[0.023,0.074]$ & $0.031$ & $[0.020,0.046]$ \\
    Wrong-singleton rate & Class 0 & $0.073$ & $[0.044,0.107]$ & $0.070$ & $[0.052,0.090]$ \\
                         & Class 1 & $0.012$ & $[0.002,0.030]$ & $0.011$ & $[0.005,0.020]$ \\
    \hline
    \end{tabular}
\end{table}

For SR-MMP, the LOO proxy and the independent audit split agree closely on the
main operational picture: most mass lies in singleton predictions for class~0,
class~1 singleton mass is smaller but still stable, and wrong-singleton rates
remain low relative to total singleton mass. This is the intended role of the
endpoint table in the paper: not to certify the LOO proxy, but to show that the
audit-based operational view yields a compact, interpretable deployment summary
at a fixed rule. A second endpoint illustrating a lower-prevalence regime is
kept in Appendix~\ref{app:tox21}.

\subsection{Solubility: scenario planning once coverage is fixed}
\label{sec:applications:solubility}

The solubility case study is a planning illustration of the attainable
operational KPI trade-offs. We train a fixed model on AquaSolDB
\citep{Sorkun2019AqSolDB}, restrict calibration to a lipophilic deployment
scenario, and use the LOO planning interface on the scenario-restricted sample
to expose the feasible operating regimes once coverage semantics have been
fixed.
Figure~\ref{fig:solubility-tradeoff} summarizes the resulting planning
interface. The left panel shows that sweeping SSBC settings does not induce a
simple monotone trade-off between soluble-class exclusion, deferral, and
decisive-correct mass; instead it yields a constrained attainable set with a
small set of Pareto-relevant regimes. The right panel shows the associated
inverse-pricing screen for a fixed downstream convention, illustrating that the
same Pareto-relevant regimes need not remain decision-optimal under the same
cost-ratio assumptions.

\begin{figure}[!b]
  \centering
  \includegraphics[width=\textwidth]{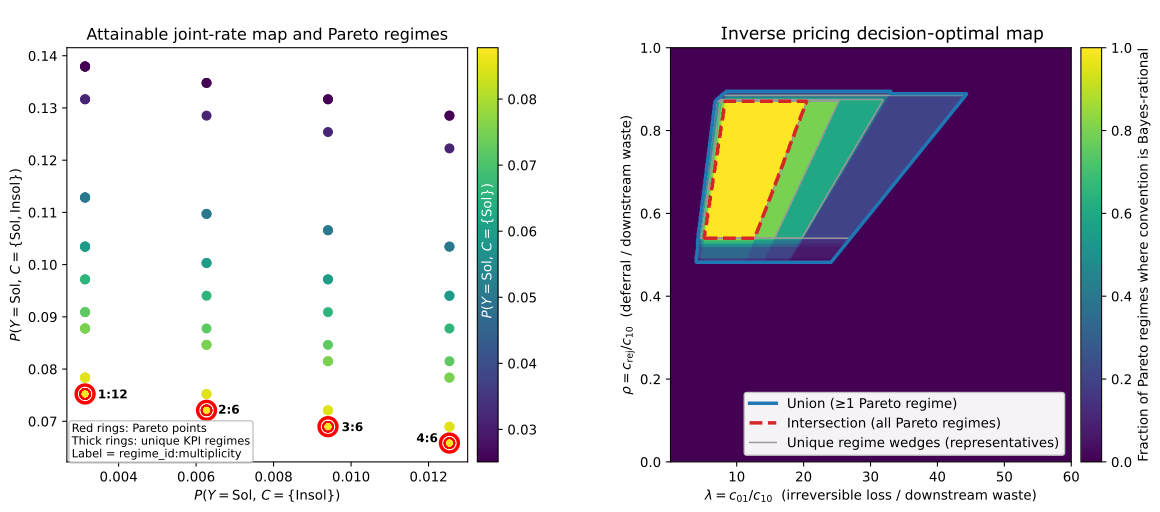}
  \caption{\textbf{Solubility planning and inverse pricing.} Left:
  attainable planning regimes induced by sweeping SSBC settings on a restricted
  deployment scenario. Right: cost-ratio regions for which a fixed downstream
  convention is decision-optimal relative to the conformal interface.}
  \label{fig:solubility-tradeoff}
\end{figure}

\FloatBarrier

The main message is qualitative: once coverage semantics are fixed, the
same conformal interface can support multiple planning regimes with materially
different operational profiles. Detailed KPI tables for selected 
Pareto-optimal regimes, scenario construction, and
parameter diagnostics are listed in Appendix~\ref{app:solubility}.

      \section{Discussion \& Conclusions}
\label{sec:discussion}
For a fixed deployed binary split-conformal rule under exchangeability,
coverage and deployment behavior are not the same object. SSBC gives a
finite-sample semantic interpretation of the coverage request, and
Calibrate--and--Audit evaluates the resulting fixed rule through the induced
region--class table. That table is the reusable deployment summary in this
setup: it supports KPI estimation, planning views, and downstream
decision-theoretic questions without retraining the base model.

This is the paper's intended division of labor. SSBC answers what coverage
claim is being made about the realized deployed rule; Calibrate--and--Audit
answers what that same fixed rule is likely to do operationally over a finite
window. The geometric analysis explains why both layers are needed: thresholds
reallocate mass across a fixed partition, so operational KPIs are coupled.

The simulations and case studies should be read at exactly that scope: the simulations
validate the intended SSBC semantics under calibration randomness, the Tox21
study shows that an independent audit split yields interpretable fixed-rule KPI
summaries in a realistic small-sample regime, and the AquaSolDB example shows
how the same interface can be used for planning once coverage semantics are
fixed. Taken together, these experiments support a division of labor that we
regard as practically useful: use SSBC to decide what coverage claim is being
made about the deployed rule, then use an audit table to understand what that
rule is likely to do operationally.

This viewpoint also sharpens a broader conceptual point. A conformal output is
not fully characterized, for deployment purposes, by its marginal coverage or
by its set-valued form alone. What matters operationally is how calibrated
regions align with labels in the deployment distribution. That is why
coverage-matched settings can still differ in commitment, deferral, decisive
error exposure, and downstream cost compatibility. In this sense, the
region--class table is not auxiliary bookkeeping; it is the minimal object in
our setup that connects a fixed conformal interface to operational and
decision-theoretic consequences. This also gives a concrete response to the HCI
concern that conformal sets can be a ``murky'' interface: the set alone is often
not enough, but the set paired with its audited region--class table makes the
deployment-facing consequences of acting on that interface explicit.

The scope remains intentionally narrow. We treat the score model as fixed,
focus on binary classification, and rely on exchangeability. The exact
finite-sample claims apply to SSBC's coverage semantics and to
independent-audit inference for fixed operational rates; the LOO interface is
used only as a planning proxy, and attainable-set sweeps and Pareto-front views
are exploratory rather than simultaneous certification statements. Extending
the same viewpoint to multi-class outputs, structured predictions, richer
abstention policies, adaptive policies, or shift-aware settings remains future
work. The most serious practical limitation is distribution or data drift:
once exchangeability fails, the finite-sample guarantees need not hold, so
deployment would require monitoring, recalibration, or shift-aware extensions 
\citep{Fannjiang2022FeedbackCovariateShift}.
Our aim is to make the fixed-rule deployment picture explicit before those
complications are layered in.

      \section*{Acknowledgements}      
      This work was supported by the U.S. Department of Energy, Office of Science,
      under Contract No.~DE-AC02-05CH11231 as part of the Laboratory Directed Research 
      and Development (LDRD) program. Additional support came in part from the U.S. Department 
      of Energy, Office of Science, Scientific Discovery through Advanced Computing (SciDAC) program 
      FORUM-AI for first-principles calculations and AI model development.

      \clearpage
      \bibliographystyle{tmlr}
      \bibliography{ssbc}
            
      \clearpage

      \appendix
\section{Appendix Outline \& Notation}
\label{app:notation}
This appendix serves as a notation guide and roadmap for the technical
appendices that follow. The appendices are, in order:
\begin{itemize}
  \item \textbf{A. Appendix Outline \& Notation} (this appendix): roadmap and
  shared notation for calibration, SSBC, and operational analysis.
  \item \textbf{B. Finite-sample distribution of calibration-conditional
  coverage:} pivot and distribution of coverage given the calibration draw;
  basis for coverage-conditional guarantees.
  \item \textbf{C. Small-Sample Beta Correction (SSBC):} mapping
  $(\alpha^\star,\delta)$ to a calibration grid index; finite-sample coverage
  semantics for the deployed rule.
  \item \textbf{D. Single-sample structural coupling, LOO decoupling, and
  planning envelopes:} why an independent audit split is needed for
  certification; leave-one-out surrogate and envelope construction when no audit
  split is available.
  \item \textbf{E. Binary conformal partitions:} four-region structure,
  regime boundary ($\tau_0+\tau_1=1$), and recovery of the argmax and Bayes 
  optimal classification rule under asymmetric costs.
  \item \textbf{F. Conformal interface-relative decision optimality and inverse
  pricing:} cost-ratio conditions for decision-optimal action; Chow-style
  accept/reject and pricing envelopes.
  \item \textbf{G. Explicit region indicators and projection masks:} formulas
  for region--label counts and projections used in auditing and KPI computation.
  \item \textbf{H. Tox21 supplementary details:} dataset, splits, and operational
  summaries for the Tox21 empirical example.
  \item \textbf{I. Solubility supplementary details:} dataset, scenario
  restriction, Pareto-front parameter roles, and $\alpha$ vs.\ $\delta$
  interpretation for the solubility case study.
\end{itemize}

\noindent
The notation below is shared across split conformal calibration, the
Small-Sample Beta Correction (SSBC), and finite-window prediction of operational
rates. The analysis mixes three layers: (i) discrete split-conformal thresholding
via order statistics on calibration scores; (ii) SSBC grid selection and
calibration-conditional coverage semantics; and (iii) predictive envelopes for
deployment-facing operational rates (Calibrate-and-Audit and a single-sample
leave-one-out (LOO) surrogate). To avoid collisions, we reserve $(k,u)$
exclusively for split conformal thresholding. Table~\ref{tab:notation} summarizes the key notation.
\clearpage

\begin{table}[htbp]
\centering
\caption{Key notation used across calibration, SSBC, and operational analysis.}
\label{tab:notation}
\footnotesize
\renewcommand{\arraystretch}{1.02}
\begin{tabular}{@{}l>{\raggedright\arraybackslash}p{3.4cm}>{\raggedright\arraybackslash}p{5.8cm}@{}}
\toprule
Symbol & Meaning & Scope / Remarks \\
\midrule
\multicolumn{3}{l}{\textit{Data splits and window sizes}} \\
\midrule
$\mathcal D_{\mathrm{cal}}$, $n_{\mathrm{cal}}$
& Calibration dataset and size
& Exchangeable sample for conformal thresholds. \\

$\mathcal D_{\mathrm{audit}}$, $n_{\mathrm{audit}}$
& Audit dataset and size
& Exchangeable sample for region--class label and operational rate estimation. \\

$m$
& Future window size
& Number of future cases for realized rates and predictive envelopes. \\

\midrule
\multicolumn{3}{l}{\textit{Region--policy--audit}} \\
\midrule
$R_\tau(x)$, $R_{\tau(\theta)}$
& Region map
& Partition of score space; $R_\tau:\mathcal X\to\{0,1\}^K$; fixed by thresholds $\tau(\theta)$. \\

$\pi$
& Deployment policy
& Maps region $R_\tau(x)$ to reported output (e.g., prediction set or $\varnothing$). \\

$\theta$, $\Theta$
& Calibration setting and space
& Indexes deployed thresholds $\tau(\theta)$ and region map $R_{\tau(\theta)}$. \\

\midrule
\multicolumn{3}{l}{\textit{Split conformal and SSBC}} \\
\midrule
$s(x,y)$, $S_i$, $S_{(k)}$
& Score, calibration scores, order statistic
& $\tau=S_{(k)}$; $k,u$ with $u=n_{\mathrm{cal}}+1-k$ index the conformal grid. \\

$\alpha_{\mathrm{grid}}$, $\alpha_{\mathrm{adj}}$
& Grid and SSBC-selected miscoverage
& $\alpha_{\mathrm{grid}}=u/(n_{\mathrm{cal}}+1)$; SSBC returns $\alpha_{\mathrm{adj}}=u^\star/(n_{\mathrm{cal}}+1)$. \\

$\alpha^\star$, $\delta$
& Target miscoverage, confidence
& User request; $\delta$ controls tail probability over calibration draws. \\

$p_{\mathrm{cov}}$, $\widehat C_m$
& Coverage probability, empirical coverage
& $p_{\mathrm{cov}}\sim\mathrm{Beta}(k,u)$; $\widehat C_m=\frac{1}{m}\sum_{j=1}^m \mathbf 1\{Y'_j\in C(X'_j)\}$. \\

\midrule
\multicolumn{3}{l}{\textit{Operational indicators and LOO}} \\
\midrule
$C(X)$, $g_\ell$, $I_\ell$
& Prediction set, event functional, indicator
& $I_\ell=g_\ell(R_{\tau(\theta)}(X),Y;\pi)$; KPIs are sums over selected region--class cells. \\

$Z_{i,j}$, $k_{\mathrm{pool},j}$, $\widehat r^{\mathrm{LOO}}_j$
& LOO indicator, pooled count, LOO rate
& $Z_{i,j}=g_j(C_{-i}(X_i),Y_i)$; $k_{\mathrm{pool},j}=\sum_i Z_{i,j}$; $\widehat r^{\mathrm{LOO}}_j=k_{\mathrm{pool},j}/n$. \\

$\texttt{infl}$, $n_{\mathrm{eff}}$
& Inflation (pessimization)
& $n_{\mathrm{eff}}=n/\texttt{infl}$ widens LOO envelopes (Appendix~\ref{app:single_sample_structural_coupling}); larger \texttt{infl} yields wider intervals. \\
\bottomrule
\end{tabular}
\end{table}

      \clearpage

\section{Finite-sample distribution of calibration-conditional coverage}
\label{app:coverage_pivot}

\noindent
This appendix derives a finite-sample characterization of the
\emph{calibration-conditional coverage} of a fixed split conformal predictor
under exchangeability. Coverage is a pure \emph{rank}
event: the future true-label score is compared to a calibration order statistic.
This yields a distribution-free pivot and a Beta law for realized
(calibration-conditional) coverage across calibration draws. This pivot is the
input to SSBC (Appendix~\ref{app:ssbc}). The derivation presented here is a 
one-shot rank/order-statistic pivot, complementary to the Beta–Binomial/de 
Finetti route in Marques Filho’s analysis of the exchangeable sequence of 
future coverage indicators \citep{Marques2025Beta}.

\subsection{Setup}

Let $\mathcal{D}_{\mathrm{cal}}=\{(X_i,Y_i)\}_{i=1}^n$ be exchangeable with a
future test pair $(X_{n+1},Y_{n+1})$. Let $s(x,y)$ be a nonconformity score and
define
\[
S_i := s(X_i,Y_i),\quad i=1,\dots,n,
\qquad\text{and}\qquad
S_{n+1}:=s(X_{n+1},Y_{n+1}).
\]
Fix an index $k\in\{1,\dots,n\}$ and set
\[
\tau := S_{(k)},
\qquad
u:=n+1-k,
\qquad
\alpha_{\mathrm{grid}}=\frac{u}{n+1}.
\]
For the predictor calibrated on $\mathcal D_{\mathrm{cal}}$, the
\emph{calibration-conditional coverage probability} is
\[
p_{\mathrm{cov}}(\mathcal D_{\mathrm{cal}})
:=
\mathbb P\!\left(S_{n+1}\le \tau \mid \mathcal D_{\mathrm{cal}}\right).
\]
This random variable varies across calibration draws but is fixed once
calibration is completed.

\subsection{Rank pivot}

Assume for clarity that the score distribution is continuous so ties occur with
probability zero (ties are addressed in Remark~\ref{rem:ties}). Under
exchangeability of $\{S_1,\dots,S_n,S_{n+1}\}$, the rank
\[
R:=\operatorname{rank}\big(S_{n+1}\ \text{among}\ S_1,\dots,S_n,S_{n+1}\big)
\]
is uniform on $\{1,\dots,n+1\}$ \citep{DavidNagaraja2003}. Since $\tau=S_{(k)}$,
\[
\{S_{n+1}\le \tau\}
\quad\Longleftrightarrow\quad
\{R\le k\}.
\]
Thus coverage is the conditional probability of a rank event.

\subsection{Beta law}

A standard order-statistic identity implies that the conditional probability of
the rank event equals the $k$th order statistic of $n+1$ i.i.d.\ uniforms:
\[
p_{\mathrm{cov}}(\mathcal D_{\mathrm{cal}})
\;\stackrel{d}{=}\;
U_{(k)},\qquad U_{(k)}\sim \mathrm{Beta}(k,u),\quad u=n+1-k,
\]
where $U_{(k)}$ is the $k$th order statistic of $U_1,\dots,U_{n+1}\overset{\mathrm{i.i.d.}}{\sim}\mathrm{Unif}(0,1)$
\citep{DavidNagaraja2003}. Equivalently, for any $t\in[0,1]$,
\[
\mathbb P\!\left(p_{\mathrm{cov}}(\mathcal D_{\mathrm{cal}})\le t\right)
=
I_t(k,u),
\]
where $I_t(\cdot,\cdot)$ is the regularized incomplete Beta function. The law is
distribution-free and depends only on $(n,k)$ (equivalently $(n,u)$).

\begin{remark}[Discrete scores and ties]
\label{rem:ties}
If scores have atoms, ranks are not almost surely unique. Exact pivots can be
recovered by randomized tie-breaking; deterministic left/right-continuous
conventions yield conservative bounds. The non-interpolated order-statistic
thresholding convention used in the main text preserves finite-sample validity.
\end{remark}

Coverage depends only on the rank of the future true-label score relative to the
calibration scores, hence admits a finite-sample distribution-free law.
Most other operational KPIs depend on the joint geometry of multiple label
scores and threshold interactions and therefore do not admit an analogous pivot
(Appendix~\ref{app:single_sample_structural_coupling}).

\section{Small-Sample Beta Correction (SSBC)}
\label{app:ssbc}

\noindent
SSBC is a deterministic index-selection rule for split conformal calibration that
makes a user request $(\alpha^\star,\delta)$ operationally precise for the
\emph{single deployed predictor} obtained after one calibration draw. SSBC
selects a conformal grid index (equivalently an order statistic) so that, with
probability at least $1-\delta$ over calibration randomness, the realized
coverage of the deployed rule is at least $1-\alpha^\star$. In the finite-window
variant, the same guarantee is imposed on empirical coverage over a future
window of size $m$. The construction relies only on the finite-sample distribution
of calibration-conditional coverage from Appendix~\ref{app:coverage_pivot}.

\subsection{Setup and conformal grid}

Let $n_{\mathrm{cal}}$ be the calibration size and let
$S_i:=s(X_i,Y_i)$ denote true-label nonconformity scores for $i=1,\dots,n_{\mathrm{cal}}$.
Split conformal selects a threshold as an order statistic:
\[
\tau = S_{(k)},\qquad k\in\{1,\dots,n_{\mathrm{cal}}\}.
\]
It is convenient to re-index the same grid by the miscoverage index
\[
u := n_{\mathrm{cal}}+1-k \in \{1,\dots,n_{\mathrm{cal}}\},
\qquad
\alpha_{\mathrm{grid}}=\frac{u}{n_{\mathrm{cal}}+1},
\qquad
k = n_{\mathrm{cal}}+1-u.
\]
Here $\alpha^\star\in(0,1)$ is the requested miscoverage and $\delta\in(0,1)$ is
a confidence/risk level controlling the probability (over calibration draws) that
the requested semantics fail.

\subsection{Distribution of realized coverage}

For the fixed predictor calibrated on $\mathcal D_{\mathrm{cal}}$, define the
calibration-conditional coverage probability
\[
p_{\mathrm{cov}}(\mathcal D_{\mathrm{cal}})
:=
\mathbb P(S_{n+1}\le \tau \mid \mathcal D_{\mathrm{cal}}),
\]
where $(X_{n+1},Y_{n+1})$ is exchangeable with the calibration sample. Under
exchangeability, Appendix~\ref{app:coverage_pivot} shows that
\[
p_{\mathrm{cov}}(\mathcal D_{\mathrm{cal}})
\sim
\mathrm{Beta}(k,u),
\qquad k=n_{\mathrm{cal}}+1-u,
\]
exactly and distribution-free. This describes how the realized coverage of
the deployed predictor varies across hypothetical recalibrations, while treating
the deployed rule as fixed after the one calibration step.

\subsection{SSBC objective: calibration-conditional PAC semantics}
\label{app:ssbc_objective}

SSBC enforces the calibration-conditional PAC-style constraint
\[
\mathbb P_{\mathcal D_{\mathrm{cal}}}\!\left(
p_{\mathrm{cov}}(\mathcal D_{\mathrm{cal}})\ \ge\ 1-\alpha^\star
\right)\ \ge\ 1-\delta.
\]
Using the Beta law, this is equivalent to
\[
\mathbb P\!\left(Z \ge 1-\alpha^\star\right)\ \ge\ 1-\delta,
\qquad
Z\sim\mathrm{Beta}(k,u).
\]

\paragraph{Selection rule: least conservative admissible grid point.}
Among discrete grid indices $u\in\{1,\dots,n_{\mathrm{cal}}\}$, SSBC selects the
\emph{largest admissible} $u$ satisfying the tail constraint above:
\[
u^\star
:=\max\Big\{u\in\{1,\dots,n_{\mathrm{cal}}\}:\ 
\mathbb P(Z\ge 1-\alpha^\star)\ge 1-\delta,\ 
Z\sim\mathrm{Beta}(n_{\mathrm{cal}}+1-u,\ u)\Big\}.
\]
Equivalently, SSBC deploys the least conservative grid miscoverage level
$\alpha_{\mathrm{adj}}=u^\star/(n_{\mathrm{cal}}+1)$ that certifies the requested
semantics. The returned order-statistic index is
$k_{\mathrm{adj}}=n_{\mathrm{cal}}+1-u^\star$.

\subsection{Finite-window deployment semantics}
\label{app:ssbc_finite_window}

In many deployments, coverage is evaluated over a finite window of size $m$.
Define empirical coverage over that window by
\[
\widehat C_m := \frac{1}{m}\sum_{j=1}^m I_{\mathrm{cov},j},
\qquad
I_{\mathrm{cov},j}:=\mathbf 1\{Y'_j \in C(X'_j)\},
\qquad
S_m := m\widehat C_m.
\]
Conditional on the calibration-conditional coverage probability
$p_{\mathrm{cov}}(\mathcal D_{\mathrm{cal}})=p$,
\[
S_m\mid p \sim \mathrm{Binomial}(m,p).
\]
Marginalizing $p\sim\mathrm{Beta}(k,u)$ yields the mixture distribution
\[
S_m \sim \mathrm{Beta\text{-}Binomial}(m; k,u),
\qquad k=n_{\mathrm{cal}}+1-u.
\]
The finite-window SSBC criterion selects $u$ such that
\[
\mathbb P\!\left(\widehat C_m \ge 1-\alpha^\star\right)\ \ge\ 1-\delta,
\]
where the probability is over both calibration randomness and the future window.

\subsubsection{Strict lower-tail convention}

Because $\widehat C_m$ is discrete, we adopt a strict violation convention
$\widehat C_m < 1-\alpha^\star$. Define the corresponding count threshold
\[
x^\star := \left\lfloor (1-\alpha^\star)m \right\rfloor + 1,
\qquad\text{so that}\qquad
\{\widehat C_m \ge 1-\alpha^\star\} \iff \{S_m \ge x^\star\}.
\]
All Beta--Binomial tail probabilities in SSBC are evaluated for the event
$\{S_m\ge x^\star\}$. This avoids boundary ambiguity when $(1-\alpha^\star)m$ is
an integer and preserves monotonicity in $u$.

\subsection{Infinite-window limit}

This limiting regime is identified by \citet{Marques2025Beta}. In our
notation, if $u_m$ denotes the SSBC-selected grid index obtained by inverting
the Beta--Binomial tail for window size $m$, and $u_\infty$ denotes the index
obtained from the Beta limit law, then $u_m \to u_\infty$ as $m\to\infty$.
Intuitively, conditional on $p_{\mathrm{cov}}=p$, the empirical coverage
$\widehat C_m$ converges almost surely to $p$, so the finite-window
Beta--Binomial criterion approaches the infinite-window Beta criterion.

\subsection{Feasibility and saturation}
\label{app:ssbc_feasibility}

Not every pair $(\alpha^\star,\delta)$ is feasible at fixed $n_{\mathrm{cal}}$
because calibration choices lie on the conformal grid. Under the most
conservative grid point $u=1$ (equivalently $k=n_{\mathrm{cal}}$),
\[
p_{\mathrm{cov}} \sim \mathrm{Beta}(n_{\mathrm{cal}},1),
\qquad
\mathbb P\!\left(p_{\mathrm{cov}} \ge 1-\alpha\right)
=
1-(1-\alpha)^{n_{\mathrm{cal}}}.
\]
Thus any infinite-window PAC requirement at confidence $1-\delta$ must satisfy
\[
\alpha \ \ge\ 1-\delta^{1/n_{\mathrm{cal}}}.
\]
If this condition is violated (and likewise in the finite-window analogue), no
grid point can satisfy the tail constraint and SSBC returns \textsc{Infeasible}.

\subsection{Coverage semantics under finite calibration}

To visualize how nominal requests map to deployed semantics under finite
calibration, Figure~\ref{fig:ssbc_semantic_map} plots the effective calibration
level $\alpha_{\mathrm{adj}}$ selected by SSBC as a function of the user inputs
$(\alpha,\delta)$ at fixed $n_{\mathrm{cal}}$. Distinct nominal requests can
induce the same deployed grid index and therefore the same realized coverage
semantics.

\begin{figure}[t]
    \centering
    \includegraphics[width=0.75\linewidth]{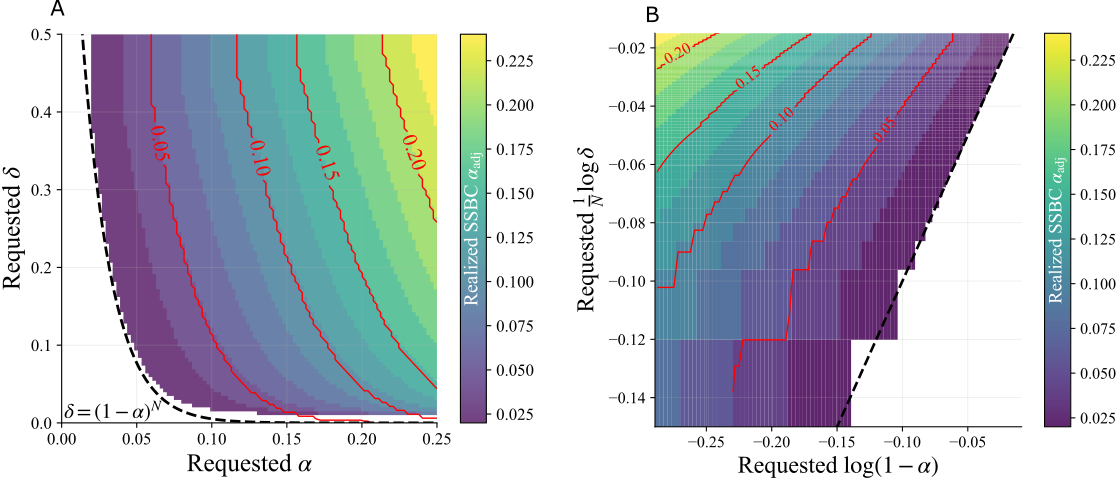}
    \caption{\textbf{Semantic interpretation of nominal coverage requests under finite calibration.}%
    Each panel visualizes the effective calibration level
    $\alpha_{\mathrm{adj}}$ selected by SSBC as a function of the user-specified
    miscoverage level $\alpha$ and confidence level $\delta$, for fixed
    $n_{\mathrm{cal}}$.
    Color encodes the deployed level $\alpha_{\mathrm{adj}}$, while contours
    indicate iso-semantic sets: distinct nominal requests that induce the same
    deployed calibration grid point.
    The feasibility boundary reflects the finite-sample constraint
    $\alpha \gtrsim 1-\delta^{1/n_{\mathrm{cal}}}$.}
    \label{fig:ssbc_semantic_map}
\end{figure}

\subsection{SSBC algorithm (deterministic specification)}
\label{app:ssbc_algorithm}

This subsection provides a reproducible implementation-level specification. The
algorithm returns the largest admissible $u$ (least conservative grid point)
satisfying the relevant tail constraint.

\begin{algorithm}[ht]
\caption{Small-Sample Beta Correction (SSBC)}
\label{alg:ssbc_detailed}
\begin{algorithmic}[1]
\Require
Target miscoverage $\alpha^\star\in(0,1)$;
calibration size $n_{\mathrm{cal}}\in\mathbb N$;
confidence $\delta\in(0,1)$;
deployment regime $\in\{\infty,m\}$ (window size $m$ if finite)
\Ensure
Adjusted grid level $\alpha_{\mathrm{adj}}$ and index $k_{\mathrm{adj}}$,
or \textsc{Infeasible}
\State $t \gets 1-\alpha^\star$
\State $u^\star \gets -1$
\If{regime $=m$}
    \State $x^\star \gets \lfloor t\,m\rfloor + 1$
\EndIf
\For{$u=1,\dots,n_{\mathrm{cal}}$}
    \State $a \gets n_{\mathrm{cal}}+1-u$, \quad $b \gets u$
    \If{regime $=\infty$}
        \State $p_{\mathrm{tail}} \gets \Pr[Z\ge t],\ Z\sim\mathrm{Beta}(a,b)$
    \Else
        \State $p_{\mathrm{tail}} \gets \Pr[X\ge x^\star],\
        X\sim\mathrm{Beta\text{-}Binomial}(m;a,b)$
    \EndIf
    \If{$p_{\mathrm{tail}}\ge 1-\delta$}
        \State $u^\star \gets u$
    \EndIf
\EndFor
\If{$u^\star<0$}
    \State \Return \textsc{Infeasible}
\EndIf
\State $\alpha_{\mathrm{adj}} \gets \tfrac{u^\star}{n_{\mathrm{cal}}+1}$
\State $k_{\mathrm{adj}} \gets n_{\mathrm{cal}}+1-u^\star$
\State \Return $\alpha_{\mathrm{adj}},\ k_{\mathrm{adj}}$
\end{algorithmic}
\end{algorithm}
\clearpage
\paragraph{Relation to DKWM-style calibration.}
DKWM-style calibration modifies the nominal grid choice to enforce conservative,
worst-case guarantees uniformly over calibration draws and distributions. SSBC
addresses a different question: it assigns a calibration-conditional PAC meaning
to a user request $(\alpha^\star,\delta)$ for the \emph{single deployed} predictor
produced after one calibration. DKWM targets uniform validity across hypothetical
recalibrations; SSBC targets admissibility of the realized rule via Beta
(or Beta--Binomial) tails.

\subsection{Extended simulation table for calibration-conditional violations}

Table~\ref{tab:calibration_conditional_violations} reports the full
calibration-size grid for the simulation summarized in
Section~\ref{sec:applications:ssbc}.

\begin{table}[htbp]
    \centering
    \caption{\textbf{Calibration-conditional coverage violation rates with theory.}
    Target miscoverage $\alpha^\star=0.10$, confidence $\delta=0.10$.
    $\alpha_{\mathrm{grid}}$ is the grid point selected on the conformal grid,
    $\alpha_{\mathrm{cont}}$ is the requested value under the DKWM correction, and
    $m_{\mathrm{cal}}$ is the calibration-window size. Results are based on
    $10^6$ calibration draws and a finite deployment window of size
    $m_{\mathrm{infer}}=100$. The \textbf{Obs} column reports
    $\mathbb P(\widehat C_m < 1-\alpha^\star)$. \textbf{Beta} reports
    $\mathbb P(p_{\mathrm{cov}} < 1-\alpha^\star)$ under
    $p_{\mathrm{cov}}\sim\mathrm{Beta}(k,u)$. \textbf{BetaBinom} reports
    $\mathbb P(\widehat C_m < 1-\alpha^\star)$ under the induced
    Beta--Binomial law.}
    \label{tab:calibration_conditional_violations}
    \small
    \begin{tabular}{r l r r r r r r}
    \hline
    $n_{\mathrm{cal}}$ & Method & $u$ & $\alpha_{\mathrm{grid}}$
    & Obs & Beta & BetaBinom & $\alpha_{\mathrm{cont}}$ \\
    \hline
     50 & None &  5 & 0.0980 & 0.3963 & 0.4312 & 0.3964 & -- \\
     50 & SSBC &  2 & 0.0392 & 0.0476 & 0.0338 & 0.0472 & -- \\
     50 & DKWM &  1 & 0.0196 & 0.0096 & 0.0052 & 0.0095 & $-0.0731$ \\
    \hline
     75 & None &  7 & 0.0921 & 0.3454 & 0.3673 & 0.3464 & -- \\
     75 & SSBC &  4 & 0.0526 & 0.0769 & 0.0504 & 0.0768 & -- \\
     75 & DKWM &  1 & 0.0132 & 0.0016 & 0.0004 & 0.0017 & $-0.0413$ \\
    \hline
    100 & None & 10 & 0.0990 & 0.4075 & 0.4513 & 0.4071 & -- \\
    100 & SSBC &  6 & 0.0594 & 0.0960 & 0.0576 & 0.0956 & -- \\
    100 & DKWM &  1 & 0.0099 & 0.0004 & 0.0000 & 0.0004 & $-0.0224$ \\
    \hline
    150 & None & 15 & 0.0993 & 0.4119 & 0.4602 & 0.4107 & -- \\
    150 & SSBC &  9 & 0.0596 & 0.0804 & 0.0307 & 0.0801 & -- \\
    150 & DKWM &  1 & 0.0066 & 0.0000 & 0.0000 & 0.0000 & $+0.0001$ \\
    \hline
    200 & None & 20 & 0.0995 & 0.4130 & 0.4655 & 0.4124 & -- \\
    200 & SSBC & 13 & 0.0647 & 0.0974 & 0.0320 & 0.0980 & -- \\
    200 & DKWM &  2 & 0.0100 & 0.0000 & 0.0000 & 0.0000 & $+0.0135$ \\
    \hline
    250 & None & 25 & 0.0996 & 0.4126 & 0.4692 & 0.4134 & -- \\
    250 & SSBC & 16 & 0.0637 & 0.0863 & 0.0175 & 0.0858 & -- \\
    250 & DKWM &  5 & 0.0199 & 0.0003 & 0.0000 & 0.0003 & $+0.0226$ \\
    \hline
    300 & None & 30 & 0.0997 & 0.4139 & 0.4719 & 0.4141 & -- \\
    300 & SSBC & 20 & 0.0664 & 0.0969 & 0.0171 & 0.0971 & -- \\
    300 & DKWM &  8 & 0.0266 & 0.0009 & 0.0000 & 0.0009 & $+0.0293$ \\
    \hline
    500 & None & 50 & 0.0998 & 0.4153 & 0.4782 & 0.4153 & -- \\
    500 & SSBC & 34 & 0.0679 & 0.0949 & 0.0049 & 0.0955 & -- \\
    500 & DKWM & 22 & 0.0439 & 0.0097 & 0.0000 & 0.0096 & $+0.0453$ \\
    \hline
    \end{tabular}
\end{table}

      \clearpage


\section{Single-sample structural coupling, leave-one-out decoupling, and planning-envelope inflation}
\label{app:single_sample_structural_coupling}

\noindent
The two-stage predictive reference used throughout this paper separates
\emph{calibration} (choose a conformal threshold on $\mathcal D_{\mathrm{cal}}$)
from \emph{operational evaluation} (estimate rates on an independent window).
This separation is what makes window indicators behave as i.i.d.\ Bernoulli draws
under a fixed deployed rule.

In practice, an independent audit set is not always available. When the same
sample is reused both to select the threshold and to estimate operational rates,
threshold selection and evaluation become coupled. This appendix records
(i) a minimal structural reason for reuse-induced dependence,
(ii) a data-efficient remedy, leave-one-out (LOO) recalibration, that provides
effective \emph{practical decoupling} in our regimes, and
(iii) an \emph{inflation} parameter \texttt{infl} that pessimizes predictive
envelopes when residual dependence or regime instability remains. Throughout,
these LOO constructions are intended for exploratory planning when a dedicated
audit split is unavailable, we cannot provide strong theoretical guarantees on the
rates of the LOO proxy indicators at this point.

\subsection{Why single-sample reuse introduces dependence}
\label{app:single_sample_structural_coupling:reuse}

Let $\mathcal D_{\mathrm{cal}}=\{(X_i,Y_i)\}_{i=1}^n$ be an exchangeable
calibration sample, let $S_i:=s(X_i,Y_i)$ be true-label nonconformity scores, and
let the split conformal threshold be the $k$th order statistic
\[
\hat\tau := S_{(k)},
\qquad k\in\{1,\dots,n\},
\]
assuming continuity so ties occur with probability zero. For any threshold $t$,
define the crossing indicator $I_i(t):=\mathbf 1\{S_i\le t\}$.

\paragraph{Insight from counting.}
If $\hat\tau=t$ is the $k$th order statistic, then exactly $k-1$ scores lie
strictly below $t$ and one score equals $t$ (under continuity, almost surely),
so
\[
\sum_{i=1}^n I_i(t) = k
\quad \text{almost surely}.
\]
This already shows that the indicators $\{I_i(t)\}_{i=1}^n$ cannot be
conditionally independent given $\hat\tau=t$: once some are known to be one,
fewer ones remain available for the others. By exchangeability, for any
$i\neq j$,
\[
\mathrm{Cov}\!\left(I_i(t),I_j(t)\mid \hat\tau=t\right)
:=\frac{k(k-n)}{n^2(n-1)} < 0
\quad (k<n).
\]

\paragraph{Implication for operational envelopes.}
Operational indicators are functions of the deployed prediction set
$C_{\hat\tau}(X)=\{y:s(X,y)\le \hat\tau\}$ and therefore inherit reuse-induced
dependence. In particular, the nonzero conditional covariance shows that these
indicators are not Bernoulli draws under a fixed rule once threshold selection
and evaluation are performed on the same sample. To recover the Bernoulli
sampling picture needed for the two-stage guarantees in the main text, one
therefore needs an independent audit set after calibration. If one na\"ively
treats reuse-based indicators as i.i.d.\ Bernoulli under a fixed rule,
predictive envelopes can become under-dispersed. 

\subsection{Approximate decoupling via leave-one-out (LOO)}
\label{app:single_sample_structural_coupling:loo}

When an independent audit sample is unavailable, we use LOO recalibration
\citep{vovk2015crossconformal,barber2021predictive} to reduce self-influence.
Let $\hat\tau_{-i}$ be the split conformal threshold computed on
$\mathcal D_{\mathrm{cal}}\setminus\{(X_i,Y_i)\}$, and let $C_{-i}(\cdot)$ be the
corresponding prediction set map. For an operational event functional
$g_j(C(X),Y)\in\{0,1\}$ (e.g., singleton, doublet, wrong-singleton), define LOO
indicators
\[
Z_{i,j} := g_j\!\left(C_{-i}(X_i),Y_i\right),\qquad i=1,\dots,n,
\]
and pooled LOO summaries
\[
k_{\mathrm{pool},j} := \sum_{i=1}^n Z_{i,j},
\qquad
\widehat r^{\mathrm{LOO}}_j := \frac{k_{\mathrm{pool},j}}{n}.
\]

\paragraph{Proxy two-stage interpretation.}
Each $Z_{i,j}$ is evaluated under a rule that does not use point $i$, restoring a
localized separation between rule construction and evaluation. Although the rule
varies across folds, each fold differs only slightly from the full calibrated
rule, so $\{Z_{i,j}\}$ can be read as indicators from nearby operating regimes.
Pooling provides a direct empirical proxy for operational behavior under finite
calibration.

\paragraph{Empirical decoupling.}
In the regimes studied (Section~\ref{sec:applications:simulations_envelopes}),
LOO envelopes track the two-stage Calibrate--and--Audit reference closely in
center and often slightly pessimistically in width. We therefore use pooled LOO
indicators for planning when a separate audit set is unavailable.

\subsection{Envelope inflation as controlled pessimization}
\label{app:single_sample_structural_coupling:infl}
LOO is local and does not remove all dependence, especially near regime
boundaries where small threshold shifts can change region support. We therefore
introduce an explicit inflation parameter $\texttt{infl}\ge 1$ that widens
predictive envelopes by shrinking an effective sample size used in the predictive
model.

\paragraph{LOO-intrinsic pessimization compounded by SSBC.}
LOO recalibration introduces an intrinsic pessimization because each fold
threshold $\hat\tau_{-i}$ is computed on only $n-1$ calibration points.
On top of this, we apply the small-sample Beta correction (SSBC) to achieve the
desired nominal level $\alpha$ under finite calibration uncertainty. The SSBC
step further increases conservatism (i.e., widens envelopes) relative to a
plug-in rate estimate, so the combined LOO$+$SSBC construction is typically
wider than either component alone. Importantly, SSBC is applied using the
\emph{actual} fold sample size (e.g., $n-1$) to correct the nominal tail level,
whereas $\texttt{infl}$ below is a separate width knob acting only through the
downstream dispersion model (we do not apply SSBC to $n_{\mathrm{eff}}$).

\paragraph{Rank stability of LOO thresholds ($\pm 1$ order statistic).}
In split conformal, the threshold is an order statistic of the calibration
scores. Let $S_1,\dots,S_n$ be the full-sample calibration scores with order
statistics $S_{(1)}\le \cdots \le S_{(n)}$, and let
$k := \lceil (n+1)(1-\alpha)\rceil$ so that $\hat\tau = S_{(k)}$.
Under LOO, the fold uses $n-1$ scores and index
$k' := \lceil n(1-\alpha)\rceil \in \{k-1,k\}$.
Removing a single score can shift the rank position of the selected order
statistic by at most one, hence (ignoring ties) the LOO threshold satisfies
\[
\hat\tau_{-i} \in \{S_{(k-1)},\, S_{(k)},\, S_{(k+1)}\},
\]
and more generally lies between adjacent full-sample order statistics. This
formalizes the ``local'' nature of LOO recalibration, while still allowing
support changes near regime boundaries when $C(\cdot)$ is sensitive to small
threshold movements.

\paragraph{Operational effect.}
In constructions below we replace a nominal proxy sample size $n$ by
$n_{\mathrm{eff}}=n/\texttt{infl}$ (and analogously for any proxy count used to
parameterize predictive dispersion). Larger \texttt{infl} yields wider envelopes
without changing the pooled mean, providing a monotone knob for conservatism.

\paragraph{Diagnostic guidance (variance-based inflation).}
As a diagnostic of calibration-induced variability, compute fold rates
$\widehat r^{(-i)}$ under each LOO-calibrated rule and their empirical variance
\[
\overline r_{\mathrm{LOO}}
=
\frac{1}{n}\sum_{i=1}^n \widehat r^{(-i)},
\qquad
\widehat{\mathrm{Var}}_{\mathrm{LOO}}
=
\frac{1}{n-1}\sum_{i=1}^n
\left(\widehat r^{(-i)} - \overline r_{\mathrm{LOO}}\right)^2.
\]
To isolate variability attributable to \emph{recalibration} (rule changes) from
simple delete-one evaluation noise, we also form a \emph{not-LOO} reference in
which the operating rule is held fixed. Let
\[
W_i := g_j\!\left(C(X_i),Y_i\right),
\qquad
\widehat r^{\mathrm{full}} := \frac{1}{n}\sum_{i=1}^n W_i,
\]
and define delete-one evaluation rates under the fixed full rule $C(\cdot)$ by
\[
\widehat r^{\mathrm{full}(-i)}
:=
\frac{1}{n-1}\sum_{\ell\neq i} W_\ell
=
\frac{n\,\widehat r^{\mathrm{full}} - W_i}{n-1},
\qquad i=1,\dots,n.
\]
Let $\widehat{\mathrm{Var}}_{\mathrm{full}}$ denote the empirical variance of
$\{\widehat r^{\mathrm{full}(-i)}\}_{i=1}^n$. We then define a variance ratio
\[
q := \frac{\widehat{\mathrm{Var}}_{\mathrm{LOO}}}{\widehat{\mathrm{Var}}_{\mathrm{full}}}
\]
and use it to suggest an inflation level
\[
\texttt{infl} := \max\{1,q\},
\]
Intuitively, $\widehat{\mathrm{Var}}_{\mathrm{full}}$ measures baseline delete-one
variability when the rule is fixed, while $\widehat{\mathrm{Var}}_{\mathrm{LOO}}$
captures additional spread induced by recalibration. Although $\hat\tau_{-i}$ can
move by at most one rank, the resulting change in region support (and therefore
in $g_j$) can be discontinuous near regime boundaries; larger $q$ flags
such sensitivity and points toward higher \texttt{infl}.

In our experiments the variance ratio is typically only slightly above $1$,
indicating that LOO recalibration contributes a modest amount of additional
variability beyond delete-one evaluation under a fixed rule. This is consistent
with the rank-$1$ stability of LOO thresholds (up to ties).

\subsection{Predictive envelope constructions from LOO proxies}
\label{app:single_sample_structural_coupling:envelopes}

We describe two complementary envelope constructions built from LOO proxy
indicators. The first mirrors the two-stage reference and is the primary
approximation; the second is a conservative guardrail.

\paragraph{Beta--Binomial planning envelopes.}
Let $Z_i$ denote a pooled proxy indicator for a fixed KPI (suppress $j$), and let
$k_{\mathrm{pool}}=\sum_{i=1}^n Z_i$ with pooled proxy rate
$\widehat p=k_{\mathrm{pool}}/n$. Apply inflation via
\[
n_{\mathrm{eff}}=\frac{n}{\texttt{infl}},
\qquad
k_{\mathrm{eff}}=\widehat p\,n_{\mathrm{eff}}=\frac{k_{\mathrm{pool}}}{\texttt{infl}}.
\]
With a small prior offset $\mathrm{offset}\in\{1,1/2\}$, define
\[
\alpha := k_{\mathrm{eff}}+\mathrm{offset},
\qquad
\beta := (n_{\mathrm{eff}}-k_{\mathrm{eff}})+\mathrm{offset}.
\]
For a future deployment window of size $m$, the predictive count $S_m$ is modeled as
\[
S_m \sim \mathrm{BetaBinomial}(m,\alpha,\beta),
\]
and equal-tailed prediction intervals follow from the Beta--Binomial CDF.
Increasing \texttt{infl} shrinks $n_{\mathrm{eff}}$ and widens intervals
monotonically while preserving the pooled mean.

\paragraph{Hoeffding-type dominance bound.}
As a conservative alternative, Hoeffding's inequality gives a distribution-free
bound for a window rate $\widehat r_m$ around a proxy mean
$\widehat r_{\mathrm{LOO}}$:
\[
\Pr\!\left(
\left| \widehat{r}_m - \widehat{r}_{\mathrm{LOO}} \right| \ge \epsilon
\right)
\le 2 \exp(-2 m \epsilon^2),
\]
yielding simple symmetric envelopes that can be used as a worst-case planning
check.

\paragraph{Summary.}
Single-sample reuse induces structural dependence through the thresholding
mechanism described above, primarily distorting predictive \emph{dispersion}.
LOO recalibration reduces self-influence and closely tracks audit-style
behavior in our regimes. Predictive envelopes are then constructed from pooled
LOO proxies using a Beta--Binomial model, with \texttt{infl} providing a
monotone knob for controlled pessimization and the Hoeffding bound serving as a
conservative guardrail. At this point, we cannot provide strong theoretical
guarantees on the rates of the LOO proxy indicators, and therefore recommend
using a proper audit set for certification.


\section{Binary conformal partitions: regimes, coupling, and rate primitives}
\label{app:geometry_binary_partition}

\noindent
This appendix records the geometric facts used in
Section~\ref{sec:setting}, Section~\ref{sec:operational-quantities}, and
Section~\ref{sec:geometric_manifolds}. Here we assume thresholds are fixed
and all probabilities are taken with respect to the deployment distribution conditional on $\mathcal
D_{\mathrm{cal}}$.

\subsection{Class-conditional split conformal as a four-region partition.}
\label{app:geometry_binary_partition:four_regions}

Let $\mathcal Y=\{0,1\}$ and let $s(x,y)$ be a nonconformity score. Class-conditional split conformal produces thresholds $\tau_0,\tau_1$, and the set-valued output is
\begin{equation}
\mathcal C(x)
=
\{0: s(x,0)\le \tau_0\}
\cup
\{1: s(x,1)\le \tau_1\},
\label{eq:mondrian_set_app_geom}
\end{equation}
equivalently represented by the region label
\[
R_\tau(x)
:=
\big(\mathbf 1\{s(x,0)\le\tau_0\},\ \mathbf 1\{s(x,1)\le\tau_1\}\big)
\in\{0,1\}^2,
\qquad \tau=(\tau_0,\tau_1).
\]
Writing $(s_0,s_1)=(s(x,0),s(x,1))$, the thresholds partition score space into
\[
R_\tau(x)=
\begin{cases}
11, & s_0\le\tau_0,\ s_1\le\tau_1 \quad \text{(doublet)},\\
10, & s_0\le\tau_0,\ s_1>\tau_1 \quad \text{(singleton $\{0\}$)},\\
01, & s_0>\tau_0,\ s_1\le\tau_1 \quad \text{(singleton $\{1\}$)},\\
00, & s_0>\tau_0,\ s_1>\tau_1 \quad \text{(abstention)}.
\end{cases}
\]
For any fixed $\tau$,
\[
\sum_{r\in\{00,01,10,11\}} \mu_r(\tau)=1,
\qquad
\mu_r(\tau):=\Pr(R_\tau(X)=r\mid\mathcal D_{\mathrm{cal}}).
\]

\subsection{Probability-normalized scores and a sharp regime boundary}
\label{app:geometry_binary_partition:feasibility_boundary}

Many probabilistic classifiers induce probability-normalized scores, e.g.
$s(x,y)=1-P(y\mid x)$. For $\mathcal Y=\{0,1\}$ this implies
\[
s(x,0)+s(x,1)=1,
\]
so feasible score pairs lie on the diagonal manifold
\[
\mathcal M=\{(u,1-u): u\in[0,1]\}.
\]
Intersecting $\mathcal M$ with the threshold rectangles yields a sharp boundary
that determines which region types can occur with nonzero mass.

\begin{proposition}[Regime boundary under probability normalization]
\label{prop:feasibility_boundary}
Assume $(s_0,s_1)\in\mathcal M$ almost surely. Then:
\begin{enumerate}
\item $R_{11}$ has nonempty intersection with $\mathcal M$ iff $\tau_0+\tau_1\ge 1$,
and has positive-length intersection iff $\tau_0+\tau_1>1$.
\item $R_{00}$ has nonempty intersection with $\mathcal M$ iff $\tau_0+\tau_1\le 1$,
and has positive-length intersection iff $\tau_0+\tau_1<1$.
\item On the boundary $\tau_0+\tau_1=1$, both $R_{11}$ and $R_{00}$ intersect
$\mathcal M$ at a single point; hence under any continuous distribution on
$\mathcal M$ they have probability zero and only the singleton regions carry
mass.
\end{enumerate}
\end{proposition}

\begin{proof}
Parameterize $\mathcal M$ by $u=s_0\in[0,1]$, so $s_1=1-u$.

$R_{11}$ requires $u\le\tau_0$ and $1-u\le\tau_1$, i.e.
$u\in[\,1-\tau_1,\ \tau_0\,]$. This interval has positive length iff
$1-\tau_1<\tau_0$, i.e.\ $\tau_0+\tau_1>1$, and degenerates to a point at equality.

$R_{00}$ requires $u>\tau_0$ and $1-u>\tau_1$, i.e.
$u\in(\,\tau_0,\ 1-\tau_1\,)$. This interval has positive length iff
$\tau_0<1-\tau_1$, i.e.\ $\tau_0+\tau_1<1$, and degenerates to a point at equality.
\end{proof}

\noindent
Crossing the affine boundary $\tau_0+\tau_1=1$ therefore removes an entire region
label (doublet or abstention) from the support of $R_\tau(X)$ under probability
normalization, explaining sharp regime changes in attainable operating behavior.

\subsection{Cross-threshold dominance in the hedging regime}
\label{app:geometry_binary_partition:cross_threshold}

Within the hedging regime $\tau_0+\tau_1>1$, the manifold $\mathcal M$ is
partitioned into three contiguous intervals corresponding to $\{10,11,01\}$. With
$u=s(x,0)$:
\[
u\in[0,\,1-\tau_1)
\Rightarrow R_\tau(x)=10,\qquad
u\in[\,1-\tau_1,\,\tau_0\,]
\Rightarrow R_\tau(x)=11,\qquad
u\in(\tau_0,\,1]
\Rightarrow R_\tau(x)=01.
\]
Hence each singleton region is controlled primarily by the \emph{opposing}
threshold: increasing $\tau_1$ expands $11$ at the expense of $10$, while
increasing $\tau_0$ expands $11$ at the expense of $01$. Operationally,
$(\tau_0,\tau_1)$ act as \emph{mass-reallocation boundaries}, not independent
per-class knobs.

Note that under probability normalization, setting thresholds to $(1-\tau_1,1-\tau_0)$ 
exchanges the roles of regions $11$ and $00$. If the deployment policy 
treats hedging and abstention identically, then this swap does not change 
action-level behavior, even though marginal coverage changes.

\subsection{Threshold classifiers from split conformal cutpoints}
\label{app:geometry_binary_partition:argmax_rule}

When $\tau_0=\tau_1=1/2$, the induced classification rule is the argmax
rule (with random tie-breaking), and coverage coincides with the
classifier's accuracy.

Using probability-normalized scores $s(x,y)=1-P(Y=y\mid X=x)$, the
threshold condition $s(x,y)\le 1/2$ is equivalent to $P(Y=y\mid X=x)\ge
1/2$. Thus, with $\tau_0=\tau_1=1/2$, the conformal set includes exactly
those labels whose posterior probability is at least $1/2$: it returns
$\{0\}$ when $P(0\mid x)>1/2$, $\{1\}$ when $P(1\mid x)>1/2$, and (only on
ties) $\{0,1\}$ when $P(0\mid x)=P(1\mid x)=1/2$. If we break ties at random
to obtain a single-label predictor, this is precisely the argmax rule.
Moreover, away from the tie set the conformal output is a singleton, so the
event $\{Y\in\mathcal C(X)\}$ is identical to $\{Y=\hat Y(X)\}$; hence the
coverage probability equals the classification accuracy (with the same
tie-breaking convention).

\paragraph{The $\tau_0+\tau_1=1$ family and asymmetric Bayes costs.}
More generally, if $\tau_0+\tau_1=1$ then abstention disappears (up to the
tie set) and the conformal output is almost surely a singleton. Indeed,
$\mathcal C(x)$ includes label $y$ iff $P(Y=y\mid X=x)\ge 1-\tau_y$, so
$\tau_0+\tau_1=1$ implies $1-\tau_0=\tau_1$ and $1-\tau_1=\tau_0$, yielding
the one-parameter threshold rule
\[
\hat Y(x)=\mathbf 1\!\{P(1\mid x)\ge \tau_0\}
        =\mathbf 1\!\{P(0\mid x)\le \tau_1\}.
\]
In this regime, $\Pr(Y\in\mathcal C(X))=\Pr(Y=\hat Y(X))$, i.e., coverage
equals the accuracy of the corresponding threshold classifier.

The parameter $\tau_0$ also has a standard decision-theoretic
interpretation. Consider asymmetric misclassification costs, with
$c_{01}$ the cost of a false positive (predicting $1$ when $Y=0$) and
$c_{10}$ the cost of a false negative (predicting $0$ when $Y=1$). The Bayes
rule that minimizes conditional expected cost predicts $1$ whenever
\[
c_{01}\,P(Y=0\mid x)\;\le\;c_{10}\,P(Y=1\mid x)
\qquad\Longleftrightarrow\qquad
P(1\mid x)\;\ge\;\frac{c_{01}}{c_{01}+c_{10}}.
\]
Thus sweeping $\tau_0$ over $(0,1)$ traces the familiar family of
cost-sensitive Bayes classifiers, with $\tau_0=c_{01}/(c_{01}+c_{10}$
corresponding to the cost ratio $c_{01}:c_{10}$.

\subsection{Region--label primitives and rate factorizations}
\label{app:geometry_binary_partition:primitives}

For auditing and planning, the primitive object is the region--class label table
\[
p_{r,y}(\tau):=\Pr(R_\tau(X)=r,\ Y=y\mid\mathcal D_{\mathrm{cal}}),
\qquad r\in\{00,01,10,11\},\ y\in\{0,1\}.
\]
Two derived summaries are
\begin{align*}
\mu_r(\tau)
&:=\Pr(R_\tau(X)=r\mid\mathcal D_{\mathrm{cal}})
=\sum_{y\in\{0,1\}}p_{r,y}(\tau),\\
\eta_r(\tau)
&:=\Pr(Y=1\mid R_\tau(X)=r,\ \mathcal D_{\mathrm{cal}})
=\frac{p_{r,1}(\tau)}{\mu_r(\tau)}\quad(\mu_r(\tau)>0).
\end{align*}
Any region-associated KPI is a projection of $\{p_{r,y}(\tau)\}$; see
Section~\ref{sec:setting:targets}.

\paragraph{Example: decisive error masses under commit-on-singletons.}
Consider the convention
\[
10\mapsto 0,\qquad 01\mapsto 1,\qquad 11,00\mapsto \text{defer / reject}.
\]
Then the decisive false-negative and false-positive masses are
\[
\mathrm{FN}_{\mathrm{dec}}(\tau)
=
\Pr(R_\tau(X)=10,\ Y=1\mid\mathcal D_{\mathrm{cal}})
=
p_{10,1}(\tau)
=
\mu_{10}(\tau)\,\eta_{10}(\tau),
\]
\[
\mathrm{FP}_{\mathrm{dec}}(\tau)
=
\Pr(R_\tau(X)=01,\ Y=0\mid\mathcal D_{\mathrm{cal}})
=
p_{01,0}(\tau)
=
\mu_{01}(\tau)\,[1-\eta_{01}(\tau)].
\]
The decisive mass is $\mu_{10}(\tau)+\mu_{01}(\tau)$, while the defer mass is
$\mu_{11}(\tau)+\mu_{00}(\tau)$ (with feasibility of $11$ versus $00$ governed by
Proposition~\ref{prop:feasibility_boundary} under probability normalization).


\section{Conformal interface-relative decision optimality and inverse pricing envelopes for a fixed conformal interface}
\label{app:cost-coherence}
\label{app:cost_pricing_envelopes}

\noindent
This appendix formalizes the operational point that distribution-free calibration
is not automatically cost-agnostic once conformal outputs are wired into actions.
We specialize classical statistical decision theory and reject-option analysis
\citep{casella2002statistical,Chow1970,Herbei2006RejectOption,Yuan2010RejectOption,Bartlett2008RejectHinge}
to the conformal interface as the deployed observable. Fix thresholds $\tau$
(calibration-conditional viewpoint) and treat the induced finite region label
$R_\tau(X)\in\mathcal R$ as the deployed observable. Any downstream rule that
uses only this interface can depend on the data only through $r\in\mathcal R$,
so whether the convention is decision-optimal relative to the conformal
interface under a given cost model is evaluated \emph{region-wise} using
within-region label frequencies (we formalize this below as
\emph{conformal interface-relative decision optimality}). The underlying concept
is simple: among actions available after observing only $R_\tau(X)$,
the deployed convention should minimize region-wise conditional expected loss. We
cast this as an \emph{inverse pricing} problem: given a fixed (interface,
convention) pair, characterize the set of consequence prices under which the
convention is decision-optimal relative to the conformal interface.

\noindent
The main point is clear: once a conformal predictor is treated as a decision
interface, decision-optimal downstream action depends on the region-wise label
frequencies, not on coverage alone and not on the set-valued output structure
alone. The contribution of this appendix is not new decision theory; it is to
make that dependence explicit for a fixed conformal interface through a worked
Chow-style case.

\subsection{Interface primitives: masses and within-region label frequencies}
\label{app:cost_pricing_envelopes:interface}

Specialize to $\mathcal Y=\{0,1\}$ and $\mathcal R=\{00,01,10,11\}$ as in
Appendix~\ref{app:geometry_binary_partition}. Adopt the calibration-conditional
joint table
\[
p_{r,y}:=\mathbb{P}(R_\tau(X)=r,\ Y=y\mid\mathcal D_{\mathrm{cal}}).
\]
Define region mass and within-region label frequency
\[
\mu_r := p_{r,0}+p_{r,1},
\qquad
\eta_r := \Pr(Y=1\mid R_\tau(X)=r,\ \mathcal D_{\mathrm{cal}})
= \frac{p_{r,1}}{\mu_r}\quad (\mu_r>0).
\]
Because $\mathcal R$ is finite, conditional uncertainty about $Y$ is piecewise
constant: within region $r$, all decision comparisons reduce to $\eta_r$.

\subsection{Region-associated action conventions and decision optimality}
\label{app:cost_pricing_envelopes:rationality}

Let $\mathcal A$ be a finite action set and let $L_\theta(a,y)$ be a priced
consequence (cost or negative utility), parameterized by $\theta\in\Theta$. A
deployed convention is a region-associated policy
\[
\tilde\pi:\mathcal R\to\mathcal A.
\]
For region $r$ with $\mu_r>0$, the interface-relative conditional risk of action
$a$ is
\[
\mathcal C_\theta(a\mid r)
:=\mathbb E[L_\theta(a,Y)\mid R_\tau(X)=r,\ \mathcal D_{\mathrm{cal}}]
=(1-\eta_r)L_\theta(a,0)+\eta_r L_\theta(a,1).
\]

\paragraph{Conformal interface-relative decision optimality.}
We say $\tilde\pi$ is \emph{decision-optimal relative to the conformal interface}
under pricing $\theta$ if for every region with $\mu_r>0$,
\[
\mathcal C_\theta(\tilde\pi(r)\mid r)
\le \mathcal C_\theta(a\mid r)
\qquad \forall a\in\mathcal A.
\]
This is precisely decision optimality relative to the coarsened observable
$R_\tau(X)$: the action chosen in each region minimizes posterior expected loss
among the admissible actions available at that interface.

\subsection{Inverse pricing envelope}
\label{app:cost_pricing_envelopes:envelope}

For each region with $\mu_r>0$, define the local feasibility set
\[
\Theta_r(\tilde\pi)
:=\Bigl\{\theta\in\Theta:\ \mathcal C_\theta(\tilde\pi(r)\mid r)\le
\mathcal C_\theta(a\mid r)\ \ \forall a\in\mathcal A\Bigr\},
\]
and the global pricing envelope
\[
\Theta(\tilde\pi):=\bigcap_{r:\mu_r>0}\Theta_r(\tilde\pi).
\]
Equivalently, $\Theta(\tilde\pi)$ is the set of price parameters for which the
forced action in each region is decision-optimal relative to the conformal
interface given only the conformal output.
Since only action comparisons matter, envelopes are naturally reported in ratio
(projective) coordinates: global positive scaling of $L_\theta$ is irrelevant,
and adding offsets $b_y$ independent of $a$ preserves comparisons.

\subsection{Worked case: Chow-style reject option and commit on singletons}
\label{app:cost_pricing_envelopes:chow_singletons}

Consider $\mathcal A=\{0,1,\mathrm{rej}\}$ with Chow-style costs
(false negative $c_{01}>0$, false positive $c_{10}>0$, rejection $c_{\mathrm{rej}}\ge 0$):
\[
L(0,1)=c_{01},\quad L(1,0)=c_{10},\quad L(\mathrm{rej},y)=c_{\mathrm{rej}},
\quad L(0,0)=L(1,1)=0.
\]
In region $r$ the conditional risks are
\[
\mathcal C(0\mid r)=\eta_r\,c_{01},\qquad
\mathcal C(1\mid r)=(1-\eta_r)\,c_{10},\qquad
\mathcal C(\mathrm{rej}\mid r)=c_{\mathrm{rej}}.
\]
Work in ratios
\[
\lambda:=\frac{c_{01}}{c_{10}},\qquad
\rho:=\frac{c_{\mathrm{rej}}}{c_{10}},
\]
so only $(\lambda,\rho)$ matters for comparisons.

\paragraph{Convention.}
\[
\tilde\pi(10)=0,\qquad \tilde\pi(01)=1,\qquad
\tilde\pi(11)=\mathrm{rej},\qquad \tilde\pi(00)=\mathrm{rej}.
\]
This convention is decision-optimal relative to the conformal interface iff the
region-wise dominance inequalities below hold for all regions with $\mu_r>0$.

\paragraph{Singleton region $01$ (output $\{1\}$).}
Choosing $1$ must beat both $0$ and $\mathrm{rej}$:
\[
(1-\eta_{01}) \le \eta_{01}\lambda
\quad\Longleftrightarrow\quad
\eta_{01}\ge \frac{1}{1+\lambda},
\qquad
(1-\eta_{01}) \le \rho
\quad\Longleftrightarrow\quad
\eta_{01}\ge 1-\rho.
\]
Thus
\[
\eta_{01}\ \ge\ \max\!\left\{\frac{1}{1+\lambda},\ 1-\rho\right\}.
\]

\paragraph{Singleton region $10$ (output $\{0\}$).}
Choosing $0$ must beat both $1$ and $\mathrm{rej}$:
\[
\eta_{10}\lambda \le (1-\eta_{10})
\quad\Longleftrightarrow\quad
\eta_{10}\le \frac{1}{1+\lambda},
\qquad
\eta_{10}\lambda \le \rho
\quad\Longleftrightarrow\quad
\eta_{10}\le \frac{\rho}{\lambda}.
\]
Thus
\[
\eta_{10}\ \le\ \min\!\left\{\frac{1}{1+\lambda},\ \frac{\rho}{\lambda}\right\}.
\]

\paragraph{Rejection regions $r\in\{11,00\}$.}
Rejecting must beat both commitments:
\[
\rho \le \eta_r\lambda
\quad\Longleftrightarrow\quad
\eta_r\ge \frac{\rho}{\lambda},
\qquad
\rho \le (1-\eta_r)
\quad\Longleftrightarrow\quad
\eta_r\le 1-\rho.
\]
Hence rejection is optimal on region $r$ only if
\[
\frac{\rho}{\lambda}\ \le\ \eta_r\ \le\ 1-\rho.
\]
The rejection band is nonempty only if
\[
\frac{\rho}{\lambda}\le 1-\rho
\quad\Longleftrightarrow\quad
\rho \le \frac{\lambda}{1+\lambda}
\quad\Longleftrightarrow\quad
c_{\mathrm{rej}} \le \frac{c_{01}c_{10}}{c_{01}+c_{10}}.
\]
\paragraph{Polyhedral summary.}
\label{app:cost_pricing_envelopes:polytope}
\label{thm:pricing_polytope}
The conditions above can be summarized as a pricing envelope in
$(\lambda,\rho)$-space. Because they are linear inequalities, the admissible
cost set is a convex polytope (intersection of half-spaces), as is standard in
reject-option decision theory \citep{Yuan2010RejectOption}. For the convention
$\tilde\pi(10)=0$, $\tilde\pi(01)=1$, $\tilde\pi(11)=\tilde\pi(00)=\mathrm{rej}$,
the envelope is the set of $(\lambda,\rho)\in\mathbb R^2_{>0}$ satisfying
\begin{align}
\lambda &\ge \frac{1-\eta_{01}}{\eta_{01}}
\quad\text{[singleton $01$: commit to $1$ beats $0$]},\\
\lambda &\le \frac{1-\eta_{10}}{\eta_{10}}
\quad\text{[singleton $10$: commit to $0$ beats $1$]},\\
\rho &\ge 1-\eta_{01}
\quad\text{[singleton $01$: commit beats reject]},\\
\rho &\ge \lambda\eta_{10}
\quad\text{[singleton $10$: commit beats reject]},\\
\rho &\le \lambda\eta_r,\quad \rho\le 1-\eta_r
\quad\text{[rejection regions $r\in\{11,00\}$ with $\mu_r>0$]},
\end{align}
together with $\lambda>0$ and $\rho\ge 0$.

\subsection{Follow-up observations}
\label{app:cost_pricing_envelopes:observations}

\paragraph{Monotone ordering.}
\label{app:cost_pricing_envelopes:nonemptiness}
\label{prop:monotone_ordering}
A useful corollary of the inequalities is that the pricing envelope is non-empty only
when the rejection regions have intermediate label composition:
$\eta_{10}\le\eta_{01}$ and, for each populated $r\in\{11,00\}$,
$\eta_{10}\le\eta_r\le\eta_{01}$. This is the region-wise analogue of the
classical reject-option condition from
\citet{Chow1970,Herbei2006RejectOption}: rejection is decision-optimal only when the
posterior falls in an intermediate band.

\paragraph{Coverage-matched settings.}
\label{app:cost_pricing_envelopes:disjoint}
\label{cor:disjoint_envelopes}
Holding the scoring model and deployment distribution fixed, different
calibration settings $\tau$ and $\tau'$ can have the same marginal coverage but
different region compositions, and therefore different pricing envelopes. In
some cases those envelopes may even fail to overlap. The point is not tied to
different datasets; it is a statement about how the same underlying interface
can behave under different calibration choices. This is the conformal-specific
implication of the preceding decision-theoretic setup: coverage-matched rules
may still be decision-optimal under different parts of cost-ratio space.
Section~\ref{sec:applications:solubility}
and Figure~\ref{fig:solubility-tradeoff} show the milder empirical version of
the same idea, namely that the feasible $(\lambda,\rho)$ region varies across
Pareto regimes.

\paragraph{Dual view.}
\label{app:cost_pricing_envelopes:dual}
The same inequalities can be read in reverse: given $(\lambda,\rho)$, which
action convention is decision-optimal relative to the conformal interface? This partitions cost-ratio space into regions
where different conventions are decision-optimal---a standard dual perspective in
cost-sensitive classification \citep{Bartlett2008RejectHinge,Yuan2010RejectOption}.
For the present appendix, this is best read as an interpretive lens on the same
region-wise quantities rather than as a separate result.

\paragraph{Summary.}
The main point of this appendix is the same one stated at the start: for a
fixed calibrated conformal predictor and a fixed wiring convention $\tilde\pi$,
decision optimality relative to the conformal interface depends on the region-wise
label frequencies through the posteriors $\{\eta_r\}$. Coverage alone is not
enough, and the bare
prediction-set structure is not enough either. The worked Chow-style case,
together with the follow-up observations on monotone ordering
(\ref{prop:monotone_ordering}) and coverage-matched settings
(\ref{cor:disjoint_envelopes}), simply makes that dependence explicit for the
conformal interface and shows how the region--class label table restricts which
cost ratios make a given downstream convention decision-optimal.


\section{Explicit region indicators and projection masks}
\label{app:explicit_region_indicators}

\noindent
This appendix briefly instantiates the linear ``sum selected cells'' formalism from
Section~\ref{sec:setting:linear} in the binary geometry used in
Figure~\ref{fig:regions} and the policy-projection example from the main text
($R_{\tau}(x)\xrightarrow{\ \pi\ }C(x)$). Region feasibility and regime facts are
recorded in Appendix~\ref{app:geometry_binary_partition}; here we keep only the
explicit projection masks needed for audit computations.

\subsection{Region labels and the region--class label table}
\label{app:explicit_region_indicators:Ptheta}

Assume $\mathcal Y=\{0,1\}$ and thresholds $\tau=(\tau_0,\tau_1)$. Define the deployed
region label
\[
R_\tau(x)=(r_0(x),r_1(x))\in\{0,1\}^2,
\qquad
r_y(x):=\mathbf 1\{s(x,y)\le \tau_y\},
\]
with region names $\mathcal R=\{r_{10},r_{11},r_{01},r_{00}\}$ as in
Appendix~\ref{app:geometry_binary_partition:four_regions}. For a calibration setting
$\theta$ (indexing deployed thresholds $\tau(\theta)$), define the calibration-conditional
region--class label table
\[
P(\theta)=\big(p_{r,y}(\theta)\big)_{r\in\mathcal R,\ y\in\mathcal Y},
\qquad
p_{r,y}(\theta)
=
\Pr\!\big(R_{\tau(\theta)}(X)=r,\ Y=y\mid \mathcal D_{\mathrm{cal}}\big).
\]
A linear operational rate is obtained by summing the subset of cells $(r,y)$ that
define the event.

\subsection{Policies used in Figure~\ref{fig:policy-projection}}
\label{app:explicit_region_indicators:policies}

A policy $\pi$ maps region labels to prediction sets $C(x)\subseteq\{0,1\}$:
\[
R_{\tau}(x)\ \xrightarrow{\ \pi\ }\ C(x).
\]
The three policies used in Figure~\ref{fig:policy-projection} are:

\paragraph{Set inclusion $\pi_{\mathrm{SI}}$.}
\[
\pi_{\mathrm{SI}}(r_{10})=\{0\},\quad
\pi_{\mathrm{SI}}(r_{11})=\{0,1\},\quad
\pi_{\mathrm{SI}}(r_{01})=\{1\},\quad
\pi_{\mathrm{SI}}(r_{00})=\varnothing.
\]

\paragraph{Commit--reject $\pi_{\mathrm{CR}}$.}
\[
\pi_{\mathrm{CR}}(r_{10})=\{0\},\quad
\pi_{\mathrm{CR}}(r_{01})=\{1\},\quad
\pi_{\mathrm{CR}}(r_{11})=\varnothing,\quad
\pi_{\mathrm{CR}}(r_{00})=\varnothing.
\]

\paragraph{Set exclusion $\pi_{\mathrm{SE}}$ (complement of set inclusion).}
\[
\pi_{\mathrm{SE}}(r_{10})=\{1\},\quad
\pi_{\mathrm{SE}}(r_{11})=\varnothing,\quad
\pi_{\mathrm{SE}}(r_{01})=\{0\},\quad
\pi_{\mathrm{SE}}(r_{00})=\{0,1\}.
\]

\subsection{Binary projection masks}
\label{app:explicit_region_indicators:masks}

Fix a policy $\pi$. For any event/quantity $\ell$, define a $4\times 2$ binary mask
\[
G_\ell(\pi)=\big(g_\ell(r,y;\pi)\big)_{r\in\mathcal R,\ y\in\mathcal Y}\in\{0,1\}^{4\times 2},
\]
where $g_\ell(r,y;\pi)=1$ indicates that the cell $(r,y)$ is included in the sum.
Then the corresponding rate is
\[
r_\ell(\theta)
=
\sum_{r\in\mathcal R}\sum_{y\in\mathcal Y} g_\ell(r,y;\pi)\,p_{r,y}(\theta).
\]
We record one detailed worked example and then list several shorter
special cases used elsewhere in the paper.

\paragraph{(1) Coverage under set inclusion.}
Under $\pi_{\mathrm{SI}}$, coverage is the event $\{Y\in C(X)\}$. Cell-by-cell, coverage holds for:
(i) $y=0$ in regions $r_{10}$ or $r_{11}$, and (ii) $y=1$ in regions $r_{01}$ or $r_{11}$. Hence
\[
G_{\mathrm{cov}}(\pi_{\mathrm{SI}})=
\begin{array}{c|cc}
 & y=0 & y=1\\ \hline
r_{10} & 1 & 0\\
r_{11} & 1 & 1\\
r_{01} & 0 & 1\\
r_{00} & 0 & 0
\end{array}
\]
and therefore
\[
r_{\mathrm{cov}}(\theta)
=
p_{10,0}(\theta)+p_{11,0}(\theta)+p_{11,1}(\theta)+p_{01,1}(\theta).
\]
Equivalently, coverage fails on the two singleton mistakes $(r_{10},1)$ and $(r_{01},0)$
and on all abstentions $r_{00}$.

\paragraph{(2) Conformal Efficiency}
In order to characterize the amount of uncertainty carried by a conformal predictor, 
the concept of conformal efficiency \citep{HEYNDRICKX2023100070} has been proposed. 
Conformal efficiency is essentially defined as the total singleton rate. The region mask 
for conformal efficiency is obtained by summing the cells that define the event:
\[
G_{\mathrm{CE}}(\pi_{\mathrm{SI}})=
\begin{array}{c|cc}
 & y=0 & y=1\\ \hline
r_{10} & 1 & 1\\
r_{11} & 0 & 0\\
r_{01} & 1 & 1\\
r_{00} & 0 & 0
\end{array}
\]
and therefore
\[
r_{\mathrm{CE}}(\theta)
=
p_{10,0}(\theta)+p_{10,1}(\theta) + p_{01,0}(\theta)+p_{01,1}(\theta).
\]

Under $\pi_{\mathrm{SI}}$, coverage holds automatically on hedged outputs ($r_{11}$),
while coverage failures occur only on singleton mistakes and abstentions. In particular,
\[
r_{\mathrm{cov}}(\theta)
=
p_{10,0}(\theta)+p_{11,0}(\theta)+p_{11,1}(\theta)+p_{01,1}(\theta),
\]
so that
\[
r_{\mathrm{cov}}(\theta)-r_{\mathrm{CE}}(\theta)
=
\big(p_{11,0}(\theta)+p_{11,1}(\theta)\big)
-
\big(p_{10,1}(\theta)+p_{01,0}(\theta)\big).
\]
Thus, when calibration places substantial mass in the hedging region $r_{11}$,
coverage can exceed conformal efficiency because hedged prediction sets are always
covered. If the calibration is performed in the abstention regime ($r_{00}$), then 
coverage and conformal efficiency differ by the total error mass among singleton predictions.

\paragraph{(3) Other projections used in the paper.}
The same mask formalism gives the remaining quantities used in
Figure~\ref{fig:hook} and in the binary operational examples:
\[
\text{missed positive mass:}\qquad
q_{10}(\theta)=\Pr(Y=1,\ C(X)=\{0\})=\Pr(Y=1,\ R_\tau(X)=r_{10})=p_{10,1}(\theta),
\]
\[
\text{hedged positive mass:}\qquad
q_{11}(\theta)=\Pr(Y=1,\ C(X)=\{0,1\})=\Pr(Y=1,\ R_\tau(X)=r_{11})=p_{11,1}(\theta),
\]
\[
\text{abstention mass under }\pi_{\mathrm{CR}}:\qquad
r_{\mathrm{abs}}(\theta;\pi_{\mathrm{CR}})
=
\big(p_{11,0}(\theta)+p_{11,1}(\theta)\big)
+
\big(p_{00,0}(\theta)+p_{00,1}(\theta)\big).
\]
These correspond respectively to a one-cell projection, another one-cell
projection, and a two-row projection. Their binary masks are obtained
immediately by placing ones on the selected cells.

\subsection{Ratios of projections (conditional diagnostics)}
\label{app:explicit_region_indicators:ratios}

Some diagnostics are conditional probabilities and therefore ratios of linear sums.
For example, the purity of singleton-$1$ outputs under $\pi_{\mathrm{SI}}$ is
\[
\mathrm{Purity}_1(\theta)
:=
\Pr\!\big(Y=1\mid C(X)=\{1\}\big)
=
\frac{\Pr(Y=1,\ C(X)=\{1\})}{\Pr(C(X)=\{1\})}
=
\frac{p_{01,1}(\theta)}{p_{01,0}(\theta)+p_{01,1}(\theta)}.
\]

\paragraph{Audit computation.}
All entries of $P(\theta)$ are estimated from region--class label counts on the audit set.
Linear KPIs are computed by summing selected cells; conditional diagnostics are computed
as ratios of such sums.

\section{Tox21 supplementary details}
\label{app:tox21}

This appendix provides supplementary documentation for the Tox21 experiments
reported in Section~\ref{sec:applications:tox21}. The purpose of this
appendix is reproducibility and contextualization rather than extension of
results or additional claims.

\subsection{Tox21 dataset}

The Tox21 benchmark consists of binary toxicity outcomes for twelve biological
assays spanning nuclear receptor (NR) signaling and cellular stress response (SR)
pathways. Each compound is labeled as \emph{active} or \emph{inactive} per assay.
Labels are sparse and highly imbalanced, particularly for nuclear receptor
targets.

The dataset composition table \ref{tab:tox21-dataset} and aggregated coverage summary \ref{tab:tox21-summary}  are recorded here.

\begin{table}[t]
    \centering
    \caption{Tox21 dataset composition and effective average class-conditional calibration sizes under the experimental protocol.}
    \label{tab:tox21-dataset}
    \footnotesize
    \begin{tabular}{lcccc}
    \toprule
    Endpoint & Total Samples & Positive Rate & Calib. Positives & Calib. Negatives \\
    \midrule
    NR-AR            & 7265 & 4.3\%  & 77  & 1739 \\
    NR-AR-LBD        & 6758 & 3.5\%  & 59  & 1630 \\
    NR-AhR           & 6549 & 11.7\% & 192 & 1445 \\
    NR-Aromatase     & 5821 & 5.2\%  & 75  & 1380 \\
    NR-ER            & 6193 & 12.8\% & 198 & 1350 \\
    NR-ER-LBD        & 6955 & 5.0\%  & 87  & 1651 \\
    NR-PPAR-$\gamma$ & 6450 & 2.9\%  & 46  & 1566 \\
    SR-ARE           & 5832 & 16.2\% & 235 & 1223 \\
    SR-ATAD5         & 7072 & 3.7\%  & 66  & 1702 \\
    SR-HSE           & 6467 & 5.8\%  & 93  & 1523 \\
    SR-MMP           & 5810 & 15.8\% & 229 & 1223 \\
    SR-p53           & 6774 & 6.2\%  & 105 & 1587 \\
    \bottomrule
    \end{tabular}
\end{table}

\begin{table}[t]
    \centering
    \caption{\textbf{Aggregated Tox21 coverage and set statistics.} Results are
    averaged over twelve endpoints and 100 random splits.}
    \label{tab:tox21-summary}
    \footnotesize
    \begin{tabular}{lcccc}
    \toprule
    Method & Mean Coverage & Violation Rate & Avg. Set Size & Singleton Rate \\
    \midrule
    Standard Split Conformal & 0.917 & 0.305 & 1.41 & 0.52 \\
    DKWM Correction          & 0.986 & 0.005 & 1.78 & 0.22 \\
    SSBC                     & 0.951 & 0.068 & 1.54 & 0.40 \\
    \bottomrule
    \end{tabular}
\end{table}

\subsection{Representative endpoint-level operational summaries}

The main text reports SR-MMP as a representative moderate-prevalence endpoint.
We keep NR-AR here to document a complementary low-prevalence regime.

\begin{table}[t]
    \centering
\caption[NR-AR operational summary]{\textbf{NR-AR endpoint.}
Operational performance summary for the NR-AR endpoint.
All reported quantities are joint probabilities normalized by the total test-set size.
Rows report the singleton rate, doublet rate, and wrong-singleton rate
(\(P(Y=c,\,|S|=1,\,\hat y \neq Y)\)) by true class.
"\(\mathcal{D}_{1}\)" and "LOO 95\% PI \(\mathcal{D}_{1}\)" denote leave-one-out point estimates and their
beta--binomial planning intervals computed on the calibration data.
"\(\mathcal{D}_{2}\)" reports empirical rates on the audit set.
"BB 95\% PI \(\mathcal{D}_{2}\)" gives Beta--Binomial predictive summaries for the
corresponding \(\mathcal{D}_{2}\) quantities over a future window of the same size.}
\label{tab:app-nr-ar-operational}
\begin{tabular}{llcccc}
    \hline
    Operational quantity & Class & $\mathcal{D}_{1}$&  LOO 95\% PI $\mathcal{D}_{1}$ &  $\mathcal{D}_{2}$ &
         BB 95\% PI $\mathcal{D}_{2}$\\
    \hline
    Singleton rate & Class 0 & $0.163$ & $[0.125,0.205]$ & $0.133$ & $[0.112,0.157]$ \\
                    & Class 1 & $0.026$ & $[0.011,0.047]$ & $0.029$ & $[0.019,0.041]$ \\
    Doublet rate & Class 0 & $0.796$ & $[0.751,0.838]$ & $0.818$ & $[0.792,0.843]$ \\
                & Class 1 & $0.015$ & $[0.005,0.031]$ & $0.020$ & $[0.012,0.030]$ \\
    Wrong-singleton rate & Class 0 & $0.086$ & $[0.058,0.119]$ & $0.060$ & $[0.045,0.077]$ \\
                & Class 1 & $0.002$ & $[0.000,0.010]$ & $0.001$ & $[0.000,0.004]$ \\
    \hline
\end{tabular}
\end{table}

Because conformal calibration is performed in a class-conditional (Mondrian)
manner, the effective calibration size for the positive class governs feasibility
and discretization effects. 

\subsection{Representation and model protocol}

All Tox21 experiments use a fixed, task-agnostic molecular representation
(descriptors plus Morgan fingerprints) and a fixed CatBoost training protocol
(1000 iterations, depth 6, learning rate 0.1, log-loss), with no assay-specific
feature engineering, class reweighting, or resampling. Molecules are parsed and
sanitized using standard cheminformatics tooling, and compounds with invalid
features are excluded. This intentionally non-optimized setup isolates conformal
calibration and operational-envelope behavior from model-architecture tuning;
observed AUROC values range between 0.80--0.85 across endpoints, and are consistent 
with baseline Tox21 classifiers.

\subsection{Conformal calibration and evaluation}

All conformal predictors are calibrated in a class-conditional (Mondrian) fashion, with separate
calibration sets for the positive and negative classes. For each run, one of
three calibration strategies is applied: standard split conformal, DKWM-based
correction, or SSBC, targeting $(\alpha,\delta)=(0.10,0.10)$.

Calibration thresholds are computed once per run and evaluated on an
independent held-out split. In the main text tables, the ``
\(\mathcal{D}_{2}\)'' columns report the corresponding endpoint-level
operational rates on that held-out split, which serve as the audit-based
reference for the operational claims. The ``\(\mathcal{D}_{1}\)'' and 
``LOO 95\% PI \(\mathcal{D}_{1}\)'' columns are
single-sample planning summaries computed from leave-one-out recalibration on
the calibration split. They are included to compare the LOO planning proxy to
the independent audit reference.

      \section{Solubility supplementary details}
\label{app:solubility}

This appendix records methodological details for the solubility
scenario-planning experiments in
Section~\ref{sec:applications:solubility}. We first describe the
dataset, partitioning, modeling, and calibration setup, and then
document the operational artifacts used to interpret the resulting
Pareto front.

\subsection{Dataset and label construction}

We use AquaSolDB as the source of aqueous solubility measurements ($\log S$),
following the curation and quality-control procedures described in the dataset
reference \citep{Sorkun2019AqSolDB}. The curated source file used in our
experiments contains 9,982 entries. Molecules with invalid SMILES strings are
excluded before feature generation. For model training and scoring we discretize
$\log S$ into three regimes with fixed thresholds: Insoluble ($\log S<-4$),
Moderate ($-4\le\log S<-2$), and Soluble ($\log S\ge-2$), defining a three-class
label space $\mathcal Y=\{0,1,2\}$ \citep{Kalepu2015Insoluble}.

\subsection{Partitioning: scaffold-based training and stratified calibration/test split}

We use a hybrid split to separate scaffold-aware model fitting from the
calibration pool used in the LOO planning analysis:

\begin{itemize}
\item \textbf{Scaffold-based training split.} Molecules are grouped by
Bemis--Murcko scaffold. The scaffold groups are shuffled with a fixed seed, and
groups are added to the training split until roughly 70\% of molecules are
assigned. This makes the training split scaffold-aware rather than purely
random.

\item \textbf{Stratified calibration/test split.} The remaining molecules are
split equally between calibration and test using stratified random sampling over
the three-class target.

\item \textbf{Calibration/test split for diagnostics.} The calibration/test
split is retained for exploratory diagnostics and documentation of the scenario
construction. The reported planning quantities in
Section~\ref{sec:applications:solubility} are not obtained from an independent
audit/test evaluation; they are obtained from the leave-one-out planning
interface on the scenario-restricted calibration sample.
\end{itemize}

\subsection{Molecular representation}

Molecules are represented using concatenated multi-resolution Morgan
fingerprints \citep{Rogers2010Morgan} computed with RDKit. We use three
fingerprint tiers: radius 4 with 512 bits, radius 3 with 1024 bits, and radius
2 with 2048 bits, yielding 3584 binary features per molecule. Molecules that
fail RDKit parsing are excluded before feature generation.

\subsection{Model and training protocol}

We train a CatBoost \citep{Prokhorenkova2018CatBoost} gradient-boosted decision
tree classifier on the scaffold-based training split with fixed hyperparameters:
1000 boosting iterations, depth 6, learning rate 0.05, random seed 42, and
multi-class loss. After training, the predictor is frozen and treated as
infrastructure; the experiments study the conformal and operational layers rather
than optimizing base accuracy.

\subsection{Calibration and operationalization}

Uncertainty quantification uses split conformal prediction with SSBC calibration
(Section~\ref{sec:coverage_control} and Appendix~\ref{app:ssbc}), which selects a
grid index to stabilize realized coverage semantics at user-specified confidence
\citep{Vovk2012PAC}. In the solubility case study, all reported planning
quantities are implemented through the leave-one-out (LOO) construction
described in Appendix~\ref{app:single_sample_structural_coupling}, rather than
through an independent audit split. This calibration layer therefore supports
the finite-window planning-envelope constructions used for operational
planning. The calibration/test split is retained as supporting context for the
scenario construction, but the reported trade-off maps and planning quantities
are generated from the LOO interface on the scenario-restricted calibration
sample.

Although the classifier is trained on three classes (Insoluble, Moderate,
Soluble), the planning analysis in Section~\ref{sec:applications:solubility}
uses a binary operational objective obtained by merging Moderate and Soluble into
a single Soluble class. Conformal prediction sets are computed in the
binary label space $\{\mathrm{Insol},\mathrm{Sol}\}$ and operational quantities
(loss, waste, hedging, decisiveness) are computed with respect to this binary
interface.

\subsection{Deployment-matched calibration via chemical tribes}
\label{sec:deployment_matched_cal}

Scenario planning conditions exchangeability on a deployment-defining event by
restricting calibration to a chemically defined subpopulation.

\paragraph{Tribe construction.}
We define coarse chemical ``tribes'' using RDKit $\mathrm{MolLogP}$:
\[
\begin{aligned}
\texttt{Tribe\_Lipophilic}:&\ \mathrm{MolLogP}>3.5,\\
\texttt{Tribe\_Hydrophilic}:&\ \mathrm{MolLogP}<1.0,\\
\texttt{Tribe\_Neutral}:&\ 1.0\le \mathrm{MolLogP}\le 3.5.
\end{aligned}
\]
We compute summary statistics of nonconformity and class probabilities
stratified by tribe to document and motivate the focal deployment regime used in
the planning study.

\paragraph{Focal deployment scenario.}
Section~\ref{sec:applications:solubility} focuses on a lipophilic deployment
regime, targeting molecules that are not naturally hydrophilic. SSBC calibration 
is therefore performed on the restricted sample of lipophilic molecules:
\[
\mathcal{D}_{\mathrm{cal}}^{\mathrm{lip}}
=
\{(X_i,Y_i)\in\mathcal{D}_{\mathrm{cal}}:\mathrm{MolLogP}(X_i)>3.5\},
\]
treated as exchangeable with respect to the intended deployment distribution.
The cutoff $\mathrm{MolLogP}>3.5$ is used here as a domain-motivated scenario
definition for lipophilic compounds, not as a tuned parameter chosen to optimize
the observed trade-off map.
After down-selection to this scenario and merging Moderate and Soluble into the
binary soluble label, the resulting planning dataset contains 319 entries:
53 soluble and 266 insoluble. These counts describe the scenario dataset over
which the planning sweep is performed; they do not correspond to a single
selected operating point.

\paragraph{Scenario-conditional exchangeability assumption.}
This restriction is scenario conditioning, not a general covariate-shift
correction. We assume deployment draws satisfy the same scenario event
$E=\{\mathrm{MolLogP}(X)>3.5\}$ so calibration and deployment are exchangeable
conditional on $E$. If deployment violates $E$, or if exchangeability fails
within $E$, conformal validity is not guaranteed. Shift-robust validity methods
are complementary and outside scope; see, e.g.,
\citet{Fannjiang2022FeedbackCovariateShift}.

\paragraph{Interpretation.}
Conditioning trades calibration sample size for improved deployment match; the
resulting loss of statistical efficiency appears directly as stronger SSBC
corrections and wider planning envelopes, consistent with transparent planning
under finite-sample uncertainty. 

\subsection{Operational outcome definitions}

Operational outcomes are defined in the main text via the binary prediction-set
interface (TP, FN, HP, FP, TN, HN). All reported quantities in the
scenario-planning analysis are joint rates over the set--label outcomes
$P(Y=y,\,C=c)$, scaled to expected counts over a window of $m=1000$ molecules.
For this case study, these quantities are estimated using pooled LOO indicators
and the associated planning-envelope construction. They are planning summaries
for the restricted scenario sample and exhibit the structure of feasible
trade-offs under the chosen scenario definition.

\subsection{Additional Pareto-front details}
\label{app:solubility_pareto}

We record the representative Pareto-optimal regimes here and then summarize the
appendix-specific details behind the solubility Pareto sweep, in particular how
the nominal SSBC parameters function along the front and how to interpret
$\alpha$ versus $\delta$ once the scenario-restricted planning interface has
been fixed.

\begin{table}[t]
    \centering
    \caption{Two representative Pareto-optimal operating regimes for the
    solubility scenario, illustrating contrasting planning postures. Rates are
    expected counts per 1000 molecules, with planning envelopes in brackets.}
    \label{tab:solubility_extremes_main}
    \begin{tabular}{l@{\hspace{1.5em}}r@{\,}r@{\hspace{1.5em}}r@{\,}r}
    \toprule
     & \multicolumn{2}{c@{\hspace{1.5em}}}{\textbf{Loss-minimizing regime}} & \multicolumn{2}{c}{\textbf{High-decisiveness regime}} \\
    \midrule
    $\alpha_0$ & \multicolumn{2}{c@{\hspace{1.5em}}}{0.125} & \multicolumn{2}{c}{0.150} \\
    $\delta_0$ & \multicolumn{2}{c@{\hspace{1.5em}}}{0.150} & \multicolumn{2}{c}{0.150} \\
    $\alpha_1$ & \multicolumn{2}{c@{\hspace{1.5em}}}{0.050} & \multicolumn{2}{c}{0.150} \\
    $\delta_1$ & \multicolumn{2}{c@{\hspace{1.5em}}}{0.075} & \multicolumn{2}{c}{0.100} \\
    \midrule
    \cmidrule(lr){2-3} \cmidrule(lr){4-5}
     & Rate & Interval & Rate & Interval \\
    \midrule
    Loss rate  &
    \makebox[1.0cm][r]{3} & \makebox[1.8cm][r]{[0,\,35]} &
    \makebox[1.0cm][r]{12} & \makebox[1.8cm][r]{[0,\,58]} \\
    Waste rate  &
    \makebox[1.0cm][r]{70} & \makebox[1.8cm][r]{[27,\,130]} &
    \makebox[1.0cm][r]{13} & \makebox[1.8cm][r]{[0,\,49]} \\
    Total hedge rate  &
    \makebox[1.0cm][r]{785} & \makebox[1.8cm][r]{[663,\,896]} &
    \makebox[1.0cm][r]{514} & \makebox[1.8cm][r]{[376,\,675]} \\
    Correct soluble singleton rate &
    \makebox[1.0cm][r]{7} & \makebox[1.8cm][r]{[0,\,38]} &
    \makebox[1.0cm][r]{88} & \makebox[1.8cm][r]{[40,\,155]} \\
    Correct insoluble singleton rate &
    \makebox[1.0cm][r]{135} & \makebox[1.8cm][r]{[68,\,215]} &
    \makebox[1.0cm][r]{282} & \makebox[1.8cm][r]{[196,\,375]} \\
    \midrule
    Decisiveness ($1000-\text{hedge}$) &
    \multicolumn{2}{c@{\hspace{1.5em}}}{\makebox[1.0cm][r]{215}} &
    \multicolumn{2}{c}{\makebox[1.0cm][r]{486}} \\
    \bottomrule
    \end{tabular}
\end{table}

\subsection{Functional roles of calibration parameters}
\label{sec:functional_roles}

The Pareto sweep in Section~\ref{sec:applications:solubility} varies the four
nominal parameters $(\alpha_0,\delta_0,\alpha_1,\delta_1)$ (with SSBC mapping them
to effective deployed grid levels). Empirically the knobs are not interchangeable:
changes reallocate mass among outcome categories (loss, waste, hedging,
decisiveness) along channels constrained by the underlying threshold geometry.

Throughout, class~0 denotes insoluble and class~1 denotes soluble.
The qualitative roles below summarize matched Pareto-optimal solutions where one
parameter varies while the others are held fixed; rates are computed from the
provided Pareto-front CSV and reported as events per 1000 molecules.

\paragraph{Role of $\alpha_0$: global conservatism against irreversible loss.}
Decreasing $\alpha_0$ suppresses loss by expanding hedging; increasing $\alpha_0$
collapses ambiguity into more decisive singleton predictions. In many matched
segments, the dominant effect is redistribution between hedged and singleton
outcomes with loss already near saturation.

\paragraph{Role of $\delta_0$: fine-scale sharpening on the insoluble side.}
For fixed $(\alpha_0,\alpha_1,\delta_1)$, $\delta_0$ tends to act locally,
shifting mass between hedged and singleton insoluble predictions with limited
impact on loss.

\paragraph{Role of $\alpha_1$: boundary-controlled tolerance with nonlocal effects.}
Although $\alpha_1$ is nominally associated with the soluble side, changing
$\alpha_1$ can move the operating point across geometric boundaries, inducing
nonlocal reallocations that may appear most strongly in insoluble outcomes.

\paragraph{Role of $\delta_1$: systematic hedge-to-insoluble reassignment.}
Across stable regions of the front, increasing $\delta_1$ often converts a
portion of hedged mass into insoluble singletons, reflecting how the normalized
score constraints and threshold geometry position the hedging interval.

\paragraph{Geometric interpretation.}
Overall, the Pareto front is shaped by geometric coupling rather than independent
tuning: each knob reallocates probability mass along constrained channels
determined by the threshold partition and the finite-sample SSBC adjustment.

\subsection{Interpreting \texorpdfstring{$\alpha$}{alpha} versus \texorpdfstring{$\delta$}{delta} on the Pareto front}
\label{sec:alpha_vs_delta}

Within SSBC the deployed threshold corresponds to an adjusted effective level
$\tilde{\alpha}$ determined jointly by $(\alpha,\delta)$. Thus $\alpha$ and
$\delta$ are best viewed as two parameterizations of a single underlying control
(the effective threshold), with different practical resolution under finite-sample
constraints. The approximate feasibility relation $\delta\approx(1-\alpha)^N$
implies
\[
\log\delta = N\log(1-\alpha),
\]
so changes in $\delta$ correspond to fine-grained (approximately logarithmic)
adjustments in $\tilde{\alpha}$, whereas changes in $\alpha$ on a coarse grid can
move the operating point across qualitatively different regions of the feasible
manifold. This helps explain why $\alpha$ sweeps can trigger large reallocations
among loss/waste/hedging outcomes, while $\delta$ sweeps more often act as local
``sharpening'' of decisions within a geometric regime.

\end{document}